\documentclass[12pt]{article}

\usepackage[psamsfonts]{amssymb}
\usepackage{amsmath, amsthm, anysize, enumerate, color, url}
\usepackage{graphicx}
\usepackage{epstopdf}
\usepackage{epsfig}
\usepackage{subfigure}

\usepackage{algorithm}
\usepackage{algpseudocode}

\usepackage{natbib}
\usepackage{threeparttable}
\usepackage{setspace}

\newtheorem{theorem}{Theorem} 

\newtheorem{corollary}{Corollary} 

\newtheorem{lemma}{Lemma} 
\newtheorem{proposition}{Proposition} 

\usepackage[colorlinks, breaklinks,
                   linkcolor=red,
                  anchorcolor=blue,
                  citecolor=blue
                  ]{hyperref}

\usepackage{breakurl}

\newtheorem{definition}{Definition}

\theoremstyle{definition}
\newtheorem{remark}{Remark}  

\newcommand{\E}{\mathbb{E}}

\newcommand{\R}{\mathbb{R}}

\renewcommand{\P}{\mathbb{P}}

\newcommand{\bs}{\boldsymbol}

\begin{document}

\def\spacingset#1{\renewcommand{\baselinestretch}%
{#1}\small\normalsize} \spacingset{1}


  \title{\bf Differentially private analysis of networks with covariates via a generalized $\beta$-model}
  \author{Ting Yan\thanks{
    Department of Statistics, Central China Normal University, Wuhan, 430079, China. \texttt{Emails}: tingyanty@mail.ccnu.edu.cn
    }}
    \date{}
  \maketitle

\bigskip
\spacingset{1.25}
\begin{abstract}
How to achieve the tradeoff between privacy and utility is one of fundamental problems
in private data analysis.
In this paper, we give a rigourous differential privacy analysis of networks in the appearance of covariates via
a generalized $\beta$-model, which has an $n$-dimensional degree parameter $\beta$
and a $p$-dimensional homophily parameter $\gamma$.
Under $(k_n, \epsilon_n)$-edge differential privacy, we use the popular Laplace mechanism to release the network statistics.
The method of moments is used to estimate the unknown model parameters.
We establish the conditions guaranteeing consistency of the differentially private estimators $\widehat{\beta}$ and $\widehat{\gamma}$
as the number of nodes $n$ goes to infinity, which reveal an
interesting tradeoff between a privacy parameter and model parameters.
The consistency is shown by applying a two-stage Newton's method to obtain the upper bound of the error between $(\widehat{\beta},\widehat{\gamma})$ and its true value $(\beta, \gamma)$
in terms of the $\ell_\infty$ distance,
which has a convergence rate of rough order $1/n^{1/2}$ for $\widehat{\beta}$ and $1/n$ for $\widehat{\gamma}$, respectively.
Further, we derive the asymptotic normalities of $\widehat{\beta}$ and $\widehat{\gamma}$, whose asymptotic variances are the same as those of
the non-private estimators under some conditions.
Our paper sheds light on
how to explore asymptotic theory under differential privacy in a principled manner;
these principled methods should be applicable to a class
of network models with covariates beyond the generalized $\beta$-model.
Numerical studies and a real data analysis demonstrate our theoretical findings.

\end{abstract}

\noindent%
\vfill

\newpage
\spacingset{1.45}

\section{Introduction}

Social network data may contain sensitive information about relationships between individuals (e.g., friendship,
email exchange, sexual interaction) and even individuals themselves (e.g., respondents in sexual partner networks).
Undoubtedly, it will expose individual's privacy if these data are directly released to the public for various research purposes.
Even if individuals are anonymized by removing identifications before being made public, it is still easy to attack by
applying some de-anonymization techniques [e.g., \cite{Narayanan:Shmatikov:2009}].  
A randomized data releasing mechanism that injects random noises to the original data (i.e., input perturbation) or their aggregate statistics
to queries (i.e., output perturbation), provides an alternative to protect data privacy.
To rigorously restrict privacy leakage,
\cite{Dwork:Mcsherry:Nissim:Smith:2006} developed a  privacy notation--\emph{differential privacy} that
requires that the output of a query does not change too much if we add/remove any single individual's record to/from a database in
randomized data releasing mechanisms. Since then, it has been widely accepted as a privacy standard for releasing sensitive data.

Many differentially private algorithms have been proposed to release network data or their aggregate statistics, especially in computer and machine learning literature
[e.g., \cite{Day:Li:Lyu:2016, MACWAN2018786, Nguyen-Imine-2016,Wang-Yan-Jiang-Leng-2022}].
On the other hand, denoising approaches have been developed to improve the estimation of
network statistics
[e.g, \cite{Hay:2009,Karwa:Slakovic:2016,Yan-2021-SS}].
However, differentially private inference in network models is still in its infancy.
This is partly because network data are nonstandard and asymptotic analysis is usually based on only one observed  network.
The increasing dimension of parameters and the appearance of noises poses additional challenge as well [\cite{fienberg2012a}].
Recently, \cite{Karwa:Slakovic:2016} derived consistency and asymptotic normality of the differentially private estimator of the parameter
constructed from the denoised degree sequence in the $\beta$-model, which is an exponential random graph model with
the degree sequence as the sufficient statistic [\cite{Chatterjee:Diaconis:Sly:2011}]. \cite{Yan-2021-SS} derived asymptotic properties of the differentially private estimators in the $p_0$ model for directed networks.
In despite of these recent developments, to the best of our knowledge,
the differentially private analyses of networks in the presence of covariates have not been explored,
in that neither their
releasing methods nor their theoretical properties are well understood.

The covariates of nodes could have important implications on the link formation.
A commonly existing phenomenon in social and econometric network data is that
individuals tend to form connections with those like themselves, which is referred to as \emph{homophily}.
Therefore, it is of interest to see how covariates influence differentially private estimation in network models.
The aim of this paper is to give a rigorously differentially private analysis of networks in the presence of covariates via
a generalized $\beta$-model.
It contains an $n$-dimensional degree parameter $\beta$ characterizing the variation in the node degrees
and a $p$-dimensional regression parameter $\gamma$ of covariates measuring phomophic or heteriphic effects.
A detailed description is given in Section \ref{subsection:model}.
This model had been introduced in  \cite{Graham2017} to model economic networks and a directed version was proposed in  \cite{Yan-Jiang-Fienberg-Leng2018}.
In this model, the degree sequence $d$ and the covariate term $y$  are the sufficient statistics.
Therefore, it is sufficient to treat only  these information as privacy contents in this model.
Under $(k_n, \epsilon_n)$-edge differential privacy, we propose to use a joint Laplace mechanism to release the network statistics $d$ and $y$,
which are added the discrete Laplacian noises and continuous Laplacian noises, respectively.

We construct estimating equations to infer model parameters based on the original maximum likelihood equations,
in which the original network statistics are directly replaced by their noisy outputs.
We develop new approaches to establish asymptotic theory of differentially private estimators.
Owning to noises having zero mean, they are the same as the moment equations. 
The main contributions are as follows.
First, we establish the conditions on a privacy parameter and model parameters
that guarantee consistency of the differentially private estimator, which
control the trade-off between privacy and utility.
A key idea for the proof is that we use a two-stage Newton's method  that first obtains the upper bound of the error
in terms of $\ell_\infty$ norm between $\widehat{\beta}_{\gamma}$ and $\beta$ with a given $\gamma$, and then derives the upper bound of the error between $\widehat{\gamma}$ and $\gamma$ by using a profiled function, where $\widehat{\beta}$ and $\widehat{\gamma}$ are the differentially private estimators of $\beta$ and $\gamma$,
respectively.
As a result, we obtain the convergence rates of $\widehat{\beta}$ and $\widehat{\gamma}$ having respective
orders of $O_p(n^{-1/2})$ and $O_p(n^{-1})$ roughly, both up to a logarithm factor.
Notably, the convergence rate for $\widehat{\beta}$ matches the minimax optimal upper bound $\|\widehat{\beta}- \beta\|_\infty =
O_p((\log p/n)^{1/2})$ for the Lasso estimator in the linear model with
$p = n$-dimensional parameter $\beta$ and the
sample size $n$ in \cite{lounici2008sup-norm}.
Second, we derive the asymptotic normal distributions of $\widehat{\beta}$ and $\widehat{\gamma}$.
This is proved by applying Taylor's expansions to a series of functions constructed from estimating equations and showing that various
remainder terms in the expansions are  asymptotically neglect.
The convergence rate $1/n$ of $\widehat{\gamma} $ makes that the asymptotic distribution
of $\widehat{\beta}$ does not depend on $\widehat{\gamma}$ and therefore has no bias.
The asymptotic distribution of $\widehat{\gamma}$ of  the homophily parameter $\gamma$ contains a bias term in terms of
a weighted sum of covariates.
Finally, we provide simulation studies as well as  a real data analysis to illustrate the theoretical results.

We note that \cite{Karwa:Slakovic:2016} obtained asymptotic results of the edge-differentially private estimator based on the denoising process
while our asymptotic results do not require the denoising process.
Another important difference from \cite{Karwa:Slakovic:2016} is that we characterize how errors of estimators depend on the privacy parameter
and we do not make the assumption that all parameters are bounded above by a constant in asymptotic theories.

For the rest of the paper, we proceed as follows.
In Section \ref{section-model-dp}, we give a necessary background on the generalized $\beta$-model and differential privacy.
In Section \ref{section-releasing-estimation}, we present the estimation.
In Section \ref{section:asymptotic}, we present the consistency and asymptotic normality of the differentially private estimator.
We carry out simulations and illustrate our results by a real data analysis in Section \ref{section:simulation}.
We give the summary and further discussion in Section \ref{section:sd}.
The proofs of the main results are regelated into Section \ref{section:appendix}.
The proofs of supported lemmas are given in the supplementary material.

\section{Model and differential privacy}
\label{section-model-dp}

In this section, we introduce the generalized $\beta$-model with covariates and present the necessary background for
differential privacy.

\subsection{Generalized $\beta$-model}
\label{subsection:model}

Let $G_n$ be an undirected graph on $n\geq 2$ nodes labeled by ``$1, \ldots, n$".
Let $A=(a_{ij})_{n\times n}$ be the adjacency matrix of $G_n$, where
$a_{ij}$ is an indicator denoting whether node $i$ is connected to node $j$.
That is, $a_{ij}=1$ if there is a link between $i$ and $j$; $a_{ij}=0$ otherwise.
We do not consider self-loops here, i.e., $a_{ii}=0$.
Let $d_i= \sum_{j \neq i} a_{ij}$ be the degree of node $i$
and $d=(d_1, \ldots, d_n)^\top$ be the degree sequence of the graph $G_n$.
We also observe a $p$-dimensional vector $z_{ij}$, the covariate information attached to the edge between nodes $i$ and $j$.
The covariate $z_{ij}$ can be formed according to the similarity or dissimilarity between
nodal attributes $z_i$ and $z_j$ for nodes $i$ and $j$. Specifically, $z_{ij}$ can be represented through a symmetric function $g(\cdot, \cdot)$ with $z_i$ and $z_j$ as its arguments. As an example if $z_{i}$ is an indicator of genders (e.g., $1$ for male and $-1$ for female),
then we could use $z_{ij} =z_i\cdot z_j$ to denote the similarity or dissimilarity measurement between $i$ and $j$.

The $\beta$-model with covariates [\cite{Graham2017,Yan-Jiang-Fienberg-Leng2018}] assumes that the edge $a_{ij}$ between $i$ and $j$
 conditional on the unobserved degree effects and observed covariates has the following probability:
\begin{equation}
\label{model}
\P( a_{ij}=a |\beta, z_{ij}) = \frac{ e^{(\beta_i + \beta_j + z_{ij}^\top \gamma)a }}{
1 +  e^{\beta_i + \beta_j + z_{ij}^\top \gamma }},~~ a\in\{0, 1\},
\end{equation}
independent of other edges.
The parameter $\beta_i$ is the intrinsic individual effect that reflects the node heterogeneity to participate in network connection.
The common parameter $\gamma$ is exogenous, measuring the homophily or heterophilic effect.
A larger homophily component $z_{ij}^\top \gamma$ means a larger homophily effect.
We will refer to $\gamma$ as the homophily parameter hereafter although it could represent heterophilic measurement.
Hereafter, we call model \eqref{model} the \emph{covariate-adjusted $\beta$-model}.

The log-likelihood function is
\begin{equation}
\label{eq-likelihood-fun}
\ell(\beta, \gamma) = \sum_{i=1}^n \beta_i d_i + \sum_{i<j} a_{ij}z_{ij}^\top \gamma
- \sum_{i<j} \log \left( 1 + \frac{e^{\beta_i + \beta_j + z_{ij}^\top \gamma }}{1+e^{\beta_i + \beta_j + z_{ij}^\top \gamma }} \right).
\end{equation}
The maximum likelihood equations are
\begin{equation}
\label{eq-MLE-equation}
\begin{array}{rcl}
d_i & = & \sum_{j\neq i} \frac{ e^{\beta_i + \beta_j + z_{ij}^\top \gamma }}{
1 +  e^{\beta_i + \beta_j + z_{ij}^\top \gamma }}, \\
\sum_{i<j} a_{ij}z_{ij}
& = &
\sum_{i<j} \frac{ z_{ij}e^{\beta_i + \beta_j + z_{ij}^\top \gamma }}{
1 +  e^{\beta_i + \beta_j + z_{ij}^\top \gamma }}.
\end{array}
\end{equation}
The R language provides a standard package ``glm" to solve \eqref{eq-MLE-equation}, which implements an iteratively reweighted least squares method for generalized linear models [\cite{McCullagh-Nelder-1989}].

\subsection{Differential privacy}
Given an original database $D$ with records of $n$ persons, we consider
a randomized data releasing mechanism $Q$ that takes $D$ as input and outputs a sanitized database $S=(S_1, \ldots, S_\ell)$ for public use.
As an illustrated example,
the additive noise mechanism returns the answer $f(D)+ z$ to the query $f(D)$, where $z$ is a
random noise. 
Let $\epsilon$  be a positive real number and $\mathcal{S}$ denote the sample space of $Q$. The data releasing mechanism $Q$ is \emph{$\epsilon$-differentially private}
if for any two neighboring databases $D_1$ and $D_2$ that differ on a single element (i.e., the data of one person),
and all measurable subsets $B$ of $\mathcal{S}$ [\cite{Dwork:Mcsherry:Nissim:Smith:2006}],
\[
Q(S\in B |D_1) \leq e^{\epsilon }\times Q(S\in B |D_2).
\]
This says the probability of an output $S$ given the input $D_1$ is less than that given the input $D_2$
multiplied by a privacy factor $e^{\epsilon}$.
The privacy parameter $\epsilon$ is chosen according to the privacy policy, which controls the trade-off
between privacy and utility. It is generally public.
Smaller value of $\epsilon$ means more privacy protection.

Differential privacy requires that the distribution of the output is almost the same whether or not an individual's record appears in the original
database.  We illustrate why it protects privacy with an example. Suppose
a hospital wants to release some statistics on the medical records of their
patients to the public. In response,
a patient may wish to make his record omitted from
the study due to a privacy concern that the published results will
reveal something about him personally.
Differential privacy alleviates this concern because
whether or not the patient participates in the study, the
probability of a possible output is almost the same.
From a theoretical point, any test statistic has nearly no power for testing
whether an individual's data is in the original database or not [\cite{Wasserman:Zhou:2010}]. 

What is being protected in the differential privacy is precisely the difference between two
neighboring databases.
Within network data, depending on the definition of
the graph neighbor, \emph{differential privacy} is divided into \emph{$k$-node differential privacy} [\cite{Hay:2009}] and
\emph{$k$-edge differential privacy} [\cite{Nissim:Raskhodnikova:Smith:2007}].
Two graphs are called neighbors if they differ in exactly $k$
edges, then \emph{differential privacy} is \emph{$k$-edge differential privacy}.
The special case with $k=1$ is generally referred to as edge differential privacy.
Analogously, we can define
\emph{$k$-node differential privacy} by letting graphs be neighbors if one can be obtained from the other by
removing $k$ nodes and its adjacent edges.
Edge differential privacy protects edges not to be detected, whereas node differential privacy protects nodes together with their
adjacent edges, which is a stronger privacy policy.
However, it may be infeasible to design algorithms that are
both node differential privacy and have good utility since it generally needs a large noise [e.g., \cite{Hay:2009}].
Following \cite{Hay:2009} and \cite{Karwa:Slakovic:2016}, we use edge differential privacy here.

Let $\delta(G, G^\prime)$ be the harming distance between two graphs $G$ and $G^\prime$,
i.e., the number of edges on which $G$ and $G^\prime$ differ.
The formal definition of $(\epsilon, k)$-edge differential privacy is as follows.

\begin{definition}[Edge differential privacy]
Let $\epsilon>0$ be a privacy parameter. A randomized mechanism
$Q(\cdot |G)$ is $(k,\epsilon)$-edge differentially private if
\[
\sup_{ G, G^\prime \in \mathcal{G}, \delta(G, G^\prime)=k } \sup_{ S\in \mathcal{S}}  \frac{ Q(S|G) }{ Q(S|G^\prime ) } \le e^\epsilon,
\]
where $\mathcal{G}$ is the set of all graphs of interest on $n$ nodes and
$\mathcal{S}$ is the set of all possible outputs.
\end{definition}

Let $f: \mathcal{G} \rightarrow \mathbb{R}^{\ell}$ be a function. The global sensitivity [\cite{Dwork:Mcsherry:Nissim:Smith:2006}] of the function $f$, denoted $\Delta f$, is defined below.

\begin{definition}
\label{definiton-2}
(Global Sensitivity).
Let $f:\mathcal{G} \to \R^\ell$. The global sensitivity of $f$ is defined as
\[
\Delta(f) = \max_{ \delta( G, G^\prime) =k } \| f(G)- f(G^\prime) \|_1
\]
where $\| \cdot \|_1$ is the $L_1$ norm.
\end{definition}

The global sensitivity measures the largest change for the query function $f$ in terms of the $L_1$-norm between any two neighboring graphs.
The magnitude of noises added in the differentially private algorithm $Q$ crucially depends on the global sensitivity.
If the outputs are the network statistics, then a simple algorithm to guarantee edge differential privacy is the Laplace Mechanism [e.g., \cite{Dwork:Mcsherry:Nissim:Smith:2006}]
that adds the Laplacian noise proportional to the global sensitivity of $f$.

\begin{lemma}\label{lemma:DLM}(Laplace Mechanism).
Suppose that $f:\mathcal{G} \to \R^\ell$ is a output function in $\mathcal{G} $. Let $z_1, \ldots, z_\ell$ be independently and identically distributed Laplace random variables with
density function $e^{-|z|/\lambda}/(2\lambda)$.
Then the Laplace mechanism outputs $f(G)+(z_1, \ldots, z_\ell)$ is $(\epsilon,k)$-edge differentially private, where $\epsilon= \Delta(f)/ \lambda$.
\end{lemma}


When $f(G)$ is integer, one can use a discrete Laplace random variable as the noise, where it has the probability mass function:
\begin{equation}
\label{equ:discrete}
\P(X=x)= \frac{1-\lambda}{1+\lambda} \lambda^{|x|},~~x \in \{0, \pm 1, \ldots\}, \lambda\in(0,1).
\end{equation}
Lemma \ref{lemma:DLM} still holds if the continuous Laplace distribution is replaced by the discrete version and the privacy parameter is chosen by
$\epsilon=-\Delta(f) \log \lambda $; see \cite{Karwa:Slakovic:2016}.

We introduce a nice property on differential privacy: any function of a differentially
private mechanism is also differentially private, as stated in the lemma below.

\begin{lemma}[\cite{Dwork:Mcsherry:Nissim:Smith:2006}]
\label{lemma:fg}
Let $f$ be an output of an $\epsilon$-differentially private mechanism and $g$ be any function. Then
$g(f(G))$ is also $\epsilon$-differentially private.
\end{lemma}

\section{Releasing network statistics and estimation}
\label{section-releasing-estimation}

\subsection{Releasing}
From the log-likelihood function \eqref{eq-likelihood-fun}, we know that $(d, \sum_{i<j} a_{ij}z_{ij})$ is the sufficient statistic.
Thus, the private information in the covariates-adjusted $\beta$--model is  $(d, \sum_{i<j} a_{ij}z_{ij})$.
We use the continuous Laplace mechanism  in Lemma \ref{lemma:DLM} and its discrete version  to
release the covariate statistic $\sum_{i<j} z_{ij}a_{ij}$ and the degree sequence $d$ under $(\epsilon_n/2, k_n)$-edge differential privacy, respectively.
The joint mechanism  satisfies $(\epsilon_n, k_n)$-edge differential privacy.
The subscript $n$ means that $k_n$ and $\epsilon_n$ are allowed to depend on $n$.
If we add $k_1$ or remove $k_2 = k_n-k_1$ edges in $G_n$ and denote the induced graph as $G_n^\prime$, then
\[
\| d - d^\prime \|_1 = 2k_n,~~ \| \sum_{i<j} ( a_{ij} - a^\prime_{ij})z_{ij} \|_1  \le pk_n z_*
\]
where $d^\prime$ is the degree sequence of $G_n^\prime$, $a^\prime_{ij}$ is the value of
edge $(i,j)$ in $G_n^\prime$ and  $z_*=\max_{ijk} |z_{ijk}|$.
So the global sensitivity is $ 2k_n$  for $d$ and $pk_n z_*$ for $\sum_{i<j} a_{ij}z_{ij})$.
We release the sufficient statistics $d$ and $y:=\sum_{i<j} a_{ij}z_{ij}$ as follows:
\begin{eqnarray}\label{eq-release}
\begin{array}{rcl}
\tilde{d}_i & = & d_i + \xi_i,~~i=1, \ldots, n, \\
\tilde{y}_t & = & \sum_{i<j} z_{ijt}a_{ij} + \eta_{t},~~t=1, \ldots, p,
\end{array}
\end{eqnarray}
where $\xi_i$, $i=1, \ldots, n$, are independently generated from the discrete Laplace distribution
with $\lambda_{n1}= e^{-\epsilon_n /(4k_n)}$,
and $\eta_t$, $t=1, \ldots, p$, are independently generated from the Laplace distribution with $\lambda_{n2}= 2pk_nz_*/\epsilon_n$.

\subsection{Estimation}
\label{section:estimation}

Write $\mu(x)=e^x/(1+e^x)$.
Define
\begin{equation}\label{definition-pi}
\pi_{ij}:=z_{ij}^\top \gamma + \beta_i + \beta_j.
\end{equation}
It is clear that $\mu(\pi_{ij})$ is the expectation of $a_{ij}$. When we emphasize the arguments $\beta$ and $\gamma$ in $\mu(\cdot)$, we write $\mu_{ij}(\beta, \gamma)$ instead of $\mu(\pi_{ij})$.
To estimate model parameters, we directly replace $d$ and $y$ in maximum likelihood equations \eqref{eq-MLE-equation} with their noisy
observed values $\tilde{d}$ and $\tilde{y}$:
\renewcommand{\arraystretch}{1.2}
\begin{equation}\label{eq:moment:dp}
\begin{array}{rcl}
\tilde{d}_i & = & \sum_{j\neq i} \mu_{ij}(\beta, \gamma),~~ i=1, \ldots, n, \\
\tilde{y} & = & \sum_{i=1}^n \sum_{j=1, j<i}^n   z_{ij} \mu_{ij}(\beta, \gamma).
\end{array}
\end{equation}
Because the expectations of the noises are zero,  the above equation are the same as the moment equations.

Let $(\widehat{\beta}, \widehat{\gamma})$ be the solution
to the equations \eqref{eq:moment:dp}.
Since $(\tilde{d}, \tilde{y})$ satisfies $(\epsilon_n, k_n)$-edge differential privacy, $(\widehat{\beta}, \widehat{\gamma})$ is also
$(\epsilon_n, k_n)$-edge differentially private according to Lemma \ref{lemma:fg}.
A two-step iterative algorithm by alternating between solving the first equation in  \eqref{eq:moment:dp} via the fixed point method in \cite{Chatterjee:Diaconis:Sly:2011} for a given $\gamma$
and solving the second equation in \eqref{eq:moment:dp} via the Newton method or the
gradient descent method, can be employed to obtain the solution.

\section{Asymptotic properties}
\label{section:asymptotic}
In this section, we present consistency and asymptotic normality of the differentially private estimator $(\widehat{\beta}, \widehat{\gamma})$.
We first introduce some notations. For a subset $C\subset \R^n$, let $C^0$ and $\overline{C}$ denote the interior and closure of $C$, respectively.
For a vector $x=(x_1, \ldots, x_n)^\top\in \R^n$, denote by $\|x\|$ for a general norm on vectors with the special cases
$\|x\|_\infty = \max_{1\le i\le n} |x_i|$ and $\|x\|_1=\sum_i |x_i|$ for the $\ell_\infty$- and $\ell_1$-norm of $x$ respectively.
Let $B(x, \epsilon)=\{y: \| x-y\|_\infty \le \epsilon\}$ be an $\epsilon$-neighborhood of $x$.
For an $n\times n$ matrix $J=(J_{ij})$, $\|J\|_\infty$ denotes the matrix norm induced by the $\ell_\infty$-norm on vectors in $\R^n$, i.e.,
\[
\|J\|_\infty = \max_{x\neq 0} \frac{ \|Jx\|_\infty }{\|x\|_\infty}
=\max_{1\le i\le n}\sum_{j=1}^n |J_{ij}|,
\]
and $\|J\|$ be a general matrix norm.
Define the matrix maximum norm: $\|J\|_{\max}=\max_{i,j}|J_{ij}|$.
We use the superscript ``*" to denote the true parameter under which the data are generated.
When there is no ambiguity, we omit the superscript ``*".
Define
\begin{equation*}\label{definition-kappa}
z_* := \max_{i,j} \| z_{ij} \|_\infty.
\end{equation*}
The notation $\sum_{j<i}$  is a shorthand for $\sum_{i=1}^n \sum_{j=1, j<i}^n$.

Recall that $\mu(x)=e^x/(1+e^x)$. Write $\mu^\prime$, $\mu^{\prime\prime}$ and $\mu^{\prime\prime\prime}$ as the first, second and third derivative of $\mu(x)$ on $x$, respectively.
A direct calculation gives that
\[
\mu^\prime(x) = \frac{e^x}{ (1+e^x)^2 },~~  \mu^{\prime\prime}(x) = \frac{e^x(1-e^x)}{ (1+e^x)^3 },~~ \mu^{\prime\prime\prime}(x) = \frac{e^x(1-4e^x+e^{2x})}{ (1+e^x)^4 }.
\]
It is easily checked that
\begin{equation}\label{eq-mu-d-upper}
|\mu^\prime(x)| \le \frac{1}{4}, ~~ |\mu^{\prime\prime}(x)| \le \frac{1}{4},~~ |\mu^{\prime\prime\prime}(x)| \le \frac{1}{4}.
\end{equation}

Let $\epsilon_{n1}$ and $\epsilon_{n2}$ be two small positive numbers.
Note that $f(x)=e^x(1+e^x)^{-2}$ is a decreasing function of $x$ when $x\ge 0$ and $f(x)=f(-x)$.
Recall that $\pi_{ij}=\beta_i + \beta_j + z_{ij}^\top \gamma$. Define
\begin{equation}\label{eq-definition-bn}
b_{n} := \sup_{\beta \in B(\beta^*, \epsilon_{n1}),\gamma\in B(\gamma^*, \epsilon_{n2})}  \max_{ i,j} \frac{(1+e^{\pi_{ij}})^2 }{ e^{\pi_{ij}} }
=O( e^{ 2\|\beta^*\|_\infty + \|\gamma^*\|_\infty }).
\end{equation}
In other words, we have
\[
\inf_{\beta \in B(\beta^*, \epsilon_{n1}),\gamma\in B(\gamma^*, \epsilon_{n2})}  \min_{ i,j} \frac{e^{\pi_{ij}} }{ (1+e^{\pi_{ij}})^2 }
\ge \frac{1}{b_n}.
\]
Note that $b_n \ge 1/4$.
When causing no confusion, we will simply write $\mu_{ij}$ stead of $\mu_{ij}(\beta, \gamma)$ for shorthand.
We will use the notations $\mu(\pi_{ij})$ and $\mu_{ij}(\beta, \gamma)$ interchangeably.
Hereafter, we assume that
the dimension $p$ of $z_{ij}$ is fixed.

\subsection{Consistency}
To derive consistency of the differentially private estimator,
let us first define a system of functions based on the estimating equations \eqref{eq:moment:dp}.
Define
\begin{equation}\label{eqn:def:F}
 F_i(\beta, \gamma)=  \sum\limits_{j=1, j\neq i}^n \mu_{ij}(\beta, \gamma) - \tilde{d}_i ,~~i=1, \ldots, n,
\end{equation}
and $F(\beta, \gamma)=(F_1(\beta, \gamma), \ldots, F_n(\beta, \gamma))^\top$.
Further, we define $F_{\gamma,i}(\beta)$ as the value of $F_i(\beta, \gamma)$ for an arbitrarily fixed $\gamma$ 
and $F_\gamma(\beta)=(F_{\gamma,1}(\beta), \ldots, F_{\gamma,n}(\beta))^\top$.
Let $\widehat{\beta}_\gamma$ be a solution to $F_\gamma(\beta)=0$.
Correspondingly, we define two functions for exploring the asymptotic behaviors of the estimator of the homophily parameter:
\begin{eqnarray}
\label{definition-Q}
Q(\beta, \gamma)= \sum_{j<i} z_{ij}  \mu_{ij}(\beta, \gamma) -  \tilde{y}, \\
\label{definition-Qc}
Q_c(\gamma)= \sum_{j<i} z_{ij}  \mu_{ij}(\widehat{\beta}_\gamma, \gamma) - \tilde{y}.
\end{eqnarray}
$Q_c(\gamma)$ could be viewed as the profile function of $Q(\beta, \gamma)$ in which the degree parameter $\beta$ is profiled out.
It is clear that
\begin{equation*}\label{equation:FQ}
F(\widehat{\beta}, \widehat{\gamma})=0,~~F_\gamma(\widehat{\beta}_\gamma)=0,~~Q(\widehat{\beta}, \widehat{\gamma})=0,~~Q_c(\widehat{\gamma})=0.
\end{equation*}

By the compound function derivation law, we have
\begin{eqnarray}\label{equ-derivation-a}
0=\frac{ \partial F_\gamma(\widehat{\beta}_\gamma) }{\partial \gamma^\top}
 = \frac{ \partial F(\widehat{\beta}_\gamma, \gamma) }{\partial \beta^\top}
  \frac{\partial \widehat{\beta}_\gamma }{\gamma^\top} + \frac{\partial F(\widehat{\beta}_\gamma, \gamma)}{\partial \gamma^\top},
  \\
  \label{equ-derivation-b}
\frac{ \partial Q_c(\gamma)}{ \partial \gamma^\top} = \frac{\partial Q(\widehat{\beta}_\gamma, \gamma)}{\partial \beta^\top}
 \frac{\partial \widehat{\beta}_\gamma }{\gamma^\top} + \frac{ \partial Q(\widehat{\beta}_\gamma, \gamma) }{ \partial \gamma^\top}.
\end{eqnarray}
By solving $\partial\widehat{\beta}_\gamma / \partial \gamma^\top$ in \eqref{equ-derivation-a}
and substituting it into \eqref{equ-derivation-b}, we get
the Jacobian matrix
$Q_c^\prime(\gamma)$ $(=\partial Q_c(\gamma)/\partial \gamma)$:
\begin{eqnarray}\label{equation:Qc-derivative}
\frac{ \partial Q_c(\gamma) }{ \partial \gamma^\top }  =
\frac{ \partial Q(\widehat{\beta}_\gamma, \gamma) }{ \partial \gamma^\top}
 - \frac{ \partial Q(\widehat{\beta}_\gamma, \gamma) }{\partial \beta^\top}
 \left[\frac{\partial F(\widehat{\beta}_\gamma,\gamma)}{\partial \beta^\top}  \right]^{-1}
\frac{\partial F(\widehat{\beta}_\gamma,\gamma)}{\partial \gamma^\top},
\end{eqnarray}
where
\[
\frac{ \partial Q(\widehat{\beta}_\gamma, \gamma) }{ \partial \gamma^\top}:=\frac{ \partial Q(\beta, \gamma) }{ \partial \gamma^\top}{\Big |}_{\beta=\widehat{\beta}_\gamma, \gamma=\gamma},
~~ \frac{\partial F(\widehat{\beta}_\gamma,\gamma)}{\partial \beta^\top}:= \frac{\partial F(\beta,\gamma)}{\partial \beta^\top}{\Big |}_{\beta=\widehat{\beta}_\gamma, \gamma=\gamma}.
\]
The asymptotic behavior of $\widehat{\gamma}$ crucially depends on the Jacobian matrix
$Q_c^\prime(\gamma)$.
Since $\widehat{\beta}_\gamma$ does not have a closed form, conditions that are directly imposed on $Q_c^\prime(\gamma)$ are not easily checked.
To derive feasible conditions, we define
\begin{equation}
\label{definition-H}
H(\beta, \gamma) = \frac{ \partial Q(\beta, \gamma) }{ \partial \gamma} - \frac{ \partial Q(\beta, \gamma) }{\partial \beta} \left[ \frac{\partial F(\beta, \gamma)}{\partial \beta} \right]^{-1}
\frac{\partial F(\beta, \gamma)}{\partial \gamma},
\end{equation}
which is a general form of $ \partial Q_c(\gamma) / \partial \gamma$.
Note that
$H(\widehat{\beta}_\gamma,\gamma)$ is the Fisher information matrix of the concentrated likelihood function $\ell_c(\gamma)$,
where the degree parameter $\beta$ is profiled out.
When $\beta\in B(\beta^*, \epsilon_{n1})$ and $b_n^2\kappa_n^2\epsilon_{n1}=o(1)$, we have the approximation:
\begin{equation*}\label{equation-H-appro}
\frac{1}{n^2} H(\beta, \gamma^*) = \frac{1}{n^2} H(\beta^*, \gamma^*) + o(1),
\end{equation*}
whose proof is given in the supplementary material.
We assume that there exists a number $\rho_n$ such that
\[
\sup_{\beta\in B(\beta^*, \epsilon_{n1})} \| H^{-1}(\beta, \gamma^*)\|_\infty \le  \frac{ \rho_n }{ n^2}.
\]
Note that the dimension of $H(\beta, \gamma)$ is fixed and every its entry is a sum $n(n-1)/2$ of terms.
If $n^{-2}H(\beta, \gamma^*)$ converges to a constant matrix, then $\rho_n$ is bounded.
Moreover, if $n^{-2}H(\beta, \gamma^*)$ is positively definite, then
\[
\rho_n = \sqrt{p}\times \sup_{\beta\in B(\beta^*, \epsilon_{n1})} 1/\lambda_{\min}(\beta),
\]
where  $\lambda_{\min}(\beta)$ is the smallest eigenvalue of $n^{-2}H(\beta, \gamma^*)$.

We use a two-stage Newton iterative sequence to show consistency.
In the first stage, we obtain an upper bound of the error between $\widehat{\beta}_{\gamma}$ and $\beta$
in terms of the $\ell_\infty$ norm for a given $\gamma$.
This is done by verifying the well-known Newton-Kantororich conditions, under which the optimal error bounds are established.
Then we derive the upper bound of the error between $\widehat{\gamma}$ and $\gamma$ by using a profiled function $Q_c(\gamma)$ constructed  from
estimating equations.
Now we formally state the consistency result.

\begin{theorem}\label{Theorem:con}
Let $\tilde{\epsilon}_n = 1 + (4k_n/\epsilon_n)(\log n/ n)^{1/2}$ and $\tau_n = b_n^3 + (k_n/\epsilon_n)b_n^3 + z_*/(\log n)^{1/2}$. If
\[
b_n^2 \tilde{\epsilon}_n = o\left( \sqrt{\frac{ n }{\log n}} \right), ~~ b_n^3 \rho_n^2 \tau_n z_*^2 = o\left(\frac{ n}{\log n}\right),
\]
then the differentially private estimator $(\widehat{\beta}, \widehat{\gamma})$ exists with probability approaching one
and is consistent in the sense that
\begin{align*}\label{Newton-convergence-rate}
\| \widehat{\gamma} - \gamma^{*} \|_\infty &=  O_p\left(
 \frac{ \rho_n \tau_n  \log n }{ n }  \right  )=o_p(1), \\
\| \widehat{\beta} - \beta^* \|_\infty &= O_p\left( b_n \tilde{\epsilon}_n \sqrt{\frac{\log n}{n}} \right)=o_p(1).
\end{align*}
\end{theorem}

The scalar factor $\tau_n$ appears due to that the magnitude of $\|Q_c(\gamma^*)\|_\infty$ is $O(\tau_n n\log n)$.
Note that the error bound of the parameter estimator in Theorem 3 of \cite{Karwa:Slakovic:2016} does not depend on the privacy parameter.
Our result here characterizes how the error bound varies on $\epsilon_n$.
In view of the above theorem, we present the consistency conditions and the error bounds under two special cases. The first case is that the parameters and covariates are bounded.
The second is that $k_n/\epsilon_n$ goes to zero.

\begin{corollary}
Assume that $\|\beta^*\|_\infty$ and $\| \gamma^* \|_\infty$ and $z_*$ are bounded above by a constant.
If $k_n/\epsilon_n = o(n/\log n)$, then
\[
\| \widehat{\gamma} - \gamma^{*} \|_\infty =  O_p\left( \frac{k_n}{\epsilon_n} \cdot
 \frac{   \log n }{ n }  \right  ), ~~
\| \widehat{\beta} - \beta^* \|_\infty = O_p\left( \sqrt{\frac{\log n}{n}} \right).
\]
\end{corollary}

\begin{corollary}\label{coro-2}
Assume that $k_n/\epsilon_n\to 0$.
If $b_n=o( (n/\log n)^{1/4})$, $\rho_n b_n^3 = o( n^{1/2}/\log n)$ and $z_*=o( (\log n)^{1/2})$, then
\[
\| \widehat{\gamma} - \gamma^{*} \|_\infty =  O_p\left(
 \frac{ \rho_n b_n^3  \log n }{ n }  \right  ),~~
\| \widehat{\beta} - \beta^* \|_\infty = O_p\left( b_n  \sqrt{\frac{\log n}{n}} \right).
\]
\end{corollary}

\begin{remark}
The condition $k_n/\epsilon_n \to 0$ means that the outputs are nearly the same as the original network statistics.
From Corollary \ref{coro-2}, we can see that the MLE of $\gamma$ has a convergence rate of order $\log n/n$.
\end{remark}

\subsection{Asymptotic normality of $\widehat{\beta}$}
The asymptotic distribution of $\widehat{\beta}$ depends crucially on the inverse of the Fisher information matrix $V$ of $\beta$.
Given $m, M>0$, we say an $n\times n$ matrix $V=(v_{ij})$ belongs to the matrix class $\mathcal{L}_{n}(m, M)$ if
$V$ is a diagonally balanced matrix with positive elements bounded by $m$ and $M$, i.e.,
\begin{equation*}\label{eq1}
\begin{array}{l}
v_{ii}=\sum_{j=1, j\neq i}^{n} v_{ij}, ~~i=1,\ldots, n, \\
0< m\le v_{ij} \le M, ~~ i,j=1,\ldots,n; i\neq j.
\end{array}
\end{equation*}
Clearly,  $F'_\gamma(\beta)$ belongs to the matrix class $\mathcal{L}(b_{n}^{-1}, 1/4)$
when $\beta\in B(\beta^*, \epsilon_{n1})$ and $\gamma\in B(\gamma^*, \epsilon_{n2})$.
We will obtain the asymptotic distribution of the estimator $\widehat{\beta}$ through
obtaining its asymptotic expression, which depends on the inverse of  $F'_\gamma( \beta )$.
However, its inverse does not have a closed form.
\cite{Yan:Zhao:Qin:2015}
proposed to approximate the inverse $V^{-1}$ of $V$ by a diagonal matrix
\begin{equation}\label{definition-S}
S=\mathrm{diag}(1/v_{11}, \ldots, 1/v_{nn}),
\end{equation}
and obtained the upper bound of the approximate error.

Note that $v_{ii}= \sum_{j\neq i} \mathrm{Var}( a_{ij})$.
By the central
limit theorem for the bounded case, as in \citeauthor{Loeve:1977} (\citeyear{Loeve:1977}, p.289), if $b_n=o(n^{1/2})$, then
$v_{ii}^{-1/2} \{d_i - \E(d_i)\}$ converges in distribution to the standard normal distribution.
When considering the asymptotic behaviors of the vector $(d_1, \ldots, d_r)$ with a fixed $r$, one could replace the degrees $d_1, \ldots, d_r$ by the independent random variables
$\tilde{d}_i=d_{i, r+1} + \ldots + d_{in}$, $i=1,\ldots,r$.
Therefore, we have the following proposition.

\begin{proposition}\label{pro:central:poisson}
If $b_n=o(n^{1/2})$,  then as $n\to\infty$: \\
(1)For any fixed $r\ge 1$,  the components of $(d_1 - \E (d_1), \ldots, d_r - \E (d_r))$ are
asymptotically independent and normally distributed with variances $v_{11}, \ldots, v_{rr}$,
respectively. \\
(2)More generally, $\sum_{i=1}^n c_i(d_i-\E(d_i))/\sqrt{v_{ii}}$ is asymptotically normally distributed with mean zero
and variance $\sum_{i=1}^\infty c_i^2$ whenever $c_1, c_2, \ldots$ are fixed constants and the latter sum is finite.
\end{proposition}

Part (2) follows from part (1) and the fact that
\[
\lim_{r\to\infty} \limsup_{t\to\infty}
\mathrm{Var}\left( \sum_{k=r+1}^n c_i \frac{ d_i - \E (d_i) }{\sqrt{v_{ii}}}\right)=0
\]
by Theorem 4.2 of \cite{Billingsley:1995}. To prove the above equation, it suffices to show that the eigenvalues of
the covariance matrix of $(d_i - \E (d_i))/v_{ii}^{1/2}$, $i=r+1, \ldots, n$ are bounded by 2 (for all $r<n$).
This is true by the well-known Perron-Frobenius theory:
if $A$ is a symmetric matrix with nonnegative elements,
then its largest eigenvalue is less than the largest value of row sums.

We apply a second order Taylor expansion to $\sum_{j\neq i} \mu_{ij}(\widehat{\beta}, \widehat{\gamma})$ to derive the asymptotic expression of $\widehat{\beta}$.
In the expansion, the first order term is the sum of $V(\widehat{\beta} - \beta)$
and $V_{\gamma\beta}(\widehat{\gamma}-\gamma)$, where $V_{\gamma\beta}=\partial F(\beta,\gamma)/\partial \gamma^\top$.
Since $V^{-1}$ does not have a closed form, we work with $S$ defined at \eqref{definition-S} to approximate it.
By Theorem \ref{Theorem:con}, $\widehat{\gamma}$ has a $n^{-1}$ convergence rate up to a logarithm factor.
This makes that the term $V_{\gamma\beta}(\widehat{\gamma}-\gamma)$ is an remainder term.
The second order term in the expansion is also asymptotically neglect.
Then we represent $\widehat{\theta}-\theta$ as the sum of
$S(d - \E d)$ and a remainder term. The central limit theorem is proved by establishing the asymptotic normality of $S( d - \E d)$ and
showing the remainder is negligible.
We formally state the central limit theorem as follows.

\begin{theorem}\label{Theorem-central-a}
Assume that the conditions in Theorem \ref{Theorem:con} hold.
If $b_n^3 \tilde{\varepsilon}_n^2 + z_* \rho_n \tau_n b_n = o( n^{1/2}/\log n)$,
then for fixed $k$ the vector $( v_{11}^{1/2}(\widehat{\beta}_1 - \beta^*), \ldots, v_{kk}^{1/2}(\widehat{\beta}_k - \beta^*_k))$
converges in distribution to the $k$-dimensional multivariate standard normal distribution.
\end{theorem}

\begin{remark}
The asymptotic variance of $\widehat{\beta}_i$ is $v_{ii}^{-1/2}$ lying between $O( b_n/n)$ and $O(1/n)$, which is the same as in the non-private estimator.
\end{remark}

\subsection{Asymptotic normality of $\widehat{\gamma}$}

Let $T_{ij}$ be an $n$-dimensional column vector with $i$th and $j$th elements ones and other elements zeros.
Define
\begin{equation*}
\begin{array}{c}
V = \frac{ \partial F(\beta^*, \gamma^*) }{ \partial \beta^\top }, ~~
V_{\gamma\beta} = \frac{ \partial Q(\beta^*, \gamma^*) }{ \partial \beta^\top}, \\
s_{ij}(\beta, \gamma) = (a_{ij}-\E a_{ij}) ( z_{ij} - V_{\gamma\beta} V^{-1} T_{ij}).
\end{array}
\end{equation*}
Note that $s_{ij}(\beta, \gamma)$, $i<j$, are independent vectors.
By direct calculations,
\[
(V_{\gamma\beta})_{kj} = \frac{ \partial Q_k(\beta^*,\gamma^*)}{\partial \beta_j } = \sum_{i \neq j} z_{ijk}\mu^\prime(\pi^*_{ij}),
\]
such that
\[
\| V_{\gamma\beta} \|_\infty \le \frac{z_*}{4}(n-1).
\]
By Lemma \ref{lemma-tight-V}, we have
\[
\| V_{\gamma\beta}V^{-1} \|_\infty \le \|V_{\gamma\beta}\|_\infty \| V^{-1} \|_\infty \le \frac{z_*(n-1)}{4} \cdot \frac{b_n(3n-4)}{2 (n-1)(n-2) } < b_n z_*.
\]
Therefore, all $s_{ij}(\beta^*, \gamma^*)$ are bounded.
Let $\bar{Q}=Q(\beta^*,\gamma^*)-\eta$ and $\bar{F}=F(\beta^*,\gamma^*)-\xi$. Note that
\[
\mathrm{Cov} (\sum\nolimits_{i<j } s_{ij}(\beta^*,\gamma^*) ) = \mathrm{Cov}( \bar{Q} - V_{\beta\gamma}^\top V^{-1} \bar{F} )=H(\beta^*,\gamma^*),
\]
where $H(\beta, \gamma)$ is defined at \eqref{definition-H}.
By the central
limit theorem for the bounded case, as in \citeauthor{Loeve:1977} (\citeyear{Loeve:1977}, p.289),
we have the following proposition.

\begin{proposition}\label{pro:th4-b}
For any nonzero fixed vector $c=(c_1, \ldots c_p)^\top$, if $ (c^\top H(\beta^*,\gamma^*) c)$ diverges, then \\
 $(c^\top H(\beta^*,\gamma^*) c )^{-1/2}c^\top\sum_{i< j}
\tilde{s}_{ij} (\beta^*, \gamma^*)$ converges in distribution to the standard normal distribution.
\end{proposition}

Let $N=n(n-1)/2$ and
\[
\bar{H}=  \lim_{n\to\infty} \frac{1}{N} H( \beta^*, \gamma^*).
\]
We assume that the above limit exists.
We briefly describe the idea of proving  asymptotic normality of $\widehat{\gamma}$.
We use a mean-value expansion to derive the explicit expression of $\widehat{\gamma}-\gamma^*$, which mainly contains
a term $Q_c(\gamma^*)$ multiplied by $\bar{H}^{-1}$.
Then by applying a third-order Taylor expansion to $Q_c(\gamma^*)$, we find that the first order term is
asymptotically normal, the second is the asymptotic bias term and the third is a remainder term.
The asymptotic distribution of $\widehat{\gamma}$ is stated below.

\begin{theorem}
\label{theorem-central-b}
Assume that the conditions in Theorem \ref{Theorem:con} hold.
If $b_{n}^4\tilde{\epsilon}_n^3z_* = o( n^{1/2}/(\log n)^{3/2})$,
then as $n$ goes to infinity, $\sqrt{N}c^\top (\widehat{\gamma} - \gamma)$ converges in distribution to
the normal distribution with mean $-c^\top\bar{H}^{-1}B_*$ and variance $c^\top  \bar{H} c$, where
\begin{equation}
\label{defintion-B-22}
B_* =  \frac{1}{\sqrt{N}} \sum_{k=1}^n \frac{ \sum_{j\neq k} z_{kj} \mu^{\prime\prime} (\pi_{kj}^*) }{ \sum_{j\neq k}  \mu^{\prime} (\pi_{kj}^*) }.
\end{equation}
\end{theorem}

\begin{remark}
We now discuss the bias term $B_*$ in \eqref{defintion-B-22}.
We assume that $z_{ij}$ is centered and independently drawn from a $p$-dimensional multivariate distribution
with bounded supports. Then,
by \citeauthor{Hoeffding:1963}'s (\citeyear{Hoeffding:1963}) inequality,
$\sum_{j\neq k} z_{kj} \mu_{kj}^{\prime\prime} (\pi_{ij}^*)=o_p( (n\log n)^{1/2})$
such that $B_*=O_p( (\log n/n)^{1/2})$. In this case, there is no bias in the limiting distribution of $\widehat{\gamma}$.
When $B_*=O(1)$, 
the confidence intervals
and the p-values of hypothesis testing constructed from $\widehat{\gamma}$ cannot achieve the nominal level without bias-correction. 
This is referred to as the well-known incidental parameter problem in econometrics literature [\cite{Neyman:Scott:1948, Graham2017}]. 
As in \cite{Dzemski2019}, we could use a simple analytical bias correction formula:
$\widehat{\gamma}_{bc} = \widehat{\gamma}- N^{-1/2}H^{-1}(\widehat{\beta},\widehat{\gamma}) \hat{B}$,
where $\widehat{B}$ is the plug-in estimate of $B_*$ by replacing
$\beta^*$ and $\gamma^*$ with their estimators $\widehat{\beta}$ and
$\widehat{\gamma}$, respectively.
\end{remark}

\section{Numerical studies}
\label{section:simulation}

\subsection{Simulations}
In this section, we evaluate the performance of the asymptotic theories through finite sizes of networks.
The parameters in the simulations are as follows.
The setting of the parameter ${\beta}^*$ took a linear form.
Specifically, we set $\beta_{i}^* = (i-1)c \log n/(n-1)$ for $i=1,\ldots,n$, where
we chose four different values for $c$, i.e., $c=0.05, 0.15, 0.3, 0.5$.
Since the conditions to guarantee asymptotic properties in theorems depend on the whole quantity $k_n/\epsilon_n$,
we set $k_n$ to be fixed (i.e., $k_n=1$) and let $\epsilon_n$ go to zero as $n$ increases.
We considered two different values for $\epsilon_n$: $\log n/n^{1/6}$ and $\log n/n^{1/4}$.
The edge covariates formed as follows.
For each node $i$, we generated two dichotomous random variables $x_{i1}$ and $x_{i2}$ from $\{1, -1\}$ with unequal probabilities $0.4$ and $0.6$
and equal probabilities, respectively. Then we set $z_{ij}=(x_{i1}*x_{j1}, x_{i2}*x_{j2})^\top$.
For the parameter $\gamma^*$, we let it be $(0.5, -0.5)^\top$.
Thus, the first positive value measures the homophily effect and the second negative value measures heterophilic effect.
We carried out simulations under two different sizes of networks: $n=100$ and $n=200$.
Each simulation was repeated $10,000$ times.

By Theorem \ref{Theorem-central-a}, $\hat{\xi}_{ij} = [\hat{\beta}_i-\hat{\beta}_j-(\beta_i^*-\beta_j^*)]/(1/\hat{v}_{ii}+1/\hat{v}_{jj})^{1/2}$
converges in distribution to the standard normal distributions, where $\hat{v}_{ii}$ is the estimate of $v_{ii}$ by replacing $\beta^*$ with $\hat{\beta}$.
We choose five special pairs $(1,2),(n/2,n/2+1)$, $(n-1,n)$, $(1,n/2)$ and $(1,n)$ for $(i,j)$.
We record the coverage probability of the $95\%$ confidence interval, the length of the confidence interval, and the frequency that the estimates do not exist.
These values for $\widehat{\gamma}$ are also reported.

{\renewcommand{\arraystretch}{1}
\begin{table}[!htp]\centering
\caption{The reported values are the coverage frequency ($\times 100\%$) for $\alpha_i-\alpha_j$ for a pair $(i,j)$ / the length of the confidence interval / the frequency ($\times 100\%$) that the estimate did not exist.}
\label{Table:a}
\vskip5pt
\begin{tabular}{ccllll}
\hline
$n$  & $(i,j)$      & $c=0.05$ & $c=0.15$ & $c=0.3$ & $c=0.5$ \\
\hline
\multicolumn{6}{c}{$\epsilon_n=\log n/n^{1/6}$}\\
\hline
100  & (1, 2)         &$ 93.98 / 1.19 / 0 $&$ 93.80 / 1.21 / 0 $&$ 93.51 / 1.27 / 0.08 $&$ 93.36 / 1.38 / 28.16 $  \\
     & (49,50)        &$ 93.91 / 1.20 / 0 $&$ 93.80 / 1.25 / 0 $&$ 93.57 / 1.43 / 0.08 $&$ 93.26 / 1.84 / 28.16 $   \\
     & (99, 100)      &$ 93.62 / 1.21 / 0 $&$ 93.69 / 1.32 / 0 $&$ 93.51 / 1.73 / 0.08 $&$ 96.60 / 2.89 / 28.16 $   \\
     & (1,50)         &$ 94.15 / 1.20 / 0 $&$ 93.88 / 1.23 / 0 $&$ 93.64 / 1.34 / 0.08 $&$ 92.93 / 1.61 / 28.16 $  \\
     & (1, 100)       &$ 93.91 / 1.20 / 0 $&$ 94.29 / 1.26 / 0 $&$ 93.53 / 1.50 / 0.08 $&$ 95.50 / 2.24 / 28.16 $ \\

200  & (1, 2)     &$ 94.37 / 0.84 / 0 $&$ 94.57 / 0.85 / 0 $&$ 94.44 / 0.90 / 0.01 $&$ 94.23 / 1.00 / 16.65 $    \\
     & (99,100)   &$ 94.19 / 0.84 / 0 $&$ 94.56 / 0.89 / 0 $&$ 94.08 / 1.05 / 0.01 $&$ 93.88 / 1.43 / 16.65 $    \\
     & (199,200)  &$ 94.80 / 0.85 / 0 $&$ 94.11 / 0.95 / 0 $&$ 94.76 / 1.34 / 0.01 $&$ 96.05 / 2.46 / 16.65 $    \\
     & (1,100)    &$ 94.27 / 0.84 / 0 $&$ 94.38 / 0.87 / 0 $&$ 94.51 / 0.98 / 0.01 $&$ 93.57 / 1.23 / 16.65 $    \\
     & (1,200)    &$ 94.34 / 0.84 / 0 $&$ 94.60 / 0.90 / 0 $&$ 94.69 / 1.14 / 0.01 $&$ 95.76 / 1.81 / 16.65 $   \\

\hline
\multicolumn{6}{c}{$\epsilon_n=\log n/n^{1/4}$}\\
\hline
100  & (1, 2)         &$ 92.61 / 1.20 / 0 $&$ 92.74 / 1.21 / 0 $&$ 91.88 / 1.27 / 0.40 $&$ 91.44 / 1.39 / 42.79 $   \\
     & (49,50)        &$ 92.48 / 1.20 / 0 $&$ 92.45 / 1.25 / 0 $&$ 91.53 / 1.43 / 0.40 $&$ 90.04 / 1.85 / 42.79 $    \\
     & (99, 100)      &$ 92.69 / 1.21 / 0 $&$ 92.40 / 1.32 / 0 $&$ 90.98 / 1.73 / 0.40 $&$ 95.30 / 2.89 / 42.79 $   \\
     & (1,50)         &$ 92.66 / 1.20 / 0 $&$ 92.67 / 1.23 / 0 $&$ 92.03 / 1.35 / 0.40 $&$ 90.81 / 1.61 / 42.79 $    \\
     & (1, 100)       &$ 92.52 / 1.20 / 0 $&$ 92.77 / 1.26 / 0 $&$ 91.19 / 1.51 / 0.40 $&$ 93.60 / 2.17 / 42.79 $   \\
200  & (1, 2)         &$ 93.74 / 0.84 / 0 $&$ 93.87 / 0.85 / 0 $&$ 93.59 / 0.90 / 0.04 $&$ 93.15 / 1.00 / 33.72 $    \\
     & (99,100)       &$ 93.47 / 0.84 / 0 $&$ 93.83 / 0.89 / 0 $&$ 93.18 / 1.05 / 0.04 $&$ 92.14 / 1.43 / 33.72 $   \\
     & (199,200)      &$ 93.95 / 0.85 / 0 $&$ 93.26 / 0.95 / 0 $&$ 92.98 / 1.34 / 0.04 $&$ 93.84 / 2.47 / 33.72 $    \\
     & (1,100)        &$ 93.76 / 0.84 / 0 $&$ 93.68 / 0.87 / 0 $&$ 93.57 / 0.98 / 0.04 $&$ 92.11 / 1.23 / 33.72 $     \\
     & (1,200)        &$ 93.55 / 0.84 / 0 $&$ 93.92 / 0.90 / 0 $&$ 93.09 / 1.13 / 0.04 $&$ 94.16 / 1.81 / 33.72 $    \\
\hline
\end{tabular}
\end{table}
}

Table \ref{Table:a} reports the simulation results for $\beta_i^* - \beta_j^*$.
The reported frequencies and lengths were conditional on the event that the estimator exists. We found that empirical coverage frequencies are very close to the nominal $95\%$ level when $c\le 0.3$ and $\epsilon_n=\log n/n^{1/6}$, and they are a little less than the nominal level when
$\epsilon_n=\log n/n^{1/4}$ and $n=100$.
As expected, the length of the confidence interval increases with $c$.
Conversely, it decreases as $n$ increases.
When $c=0.5$, the estimator failed to exist with a positive frequencies over $20\%$.
In other cases, estimates existed almost in every simulation.

Table \ref{Table:gamma} reports the median of $\widehat{\gamma}$ as well as those of the bias corrected estimator
$\widehat{\gamma}_{bc} = \widehat{\gamma}- \bar{H}^{-1} \hat{B}$.
As we can see, the bias is small and the empirical coverage frequencies for
the estimators $\widehat{\gamma}_{bc}$ are more closer to the target level $95\%$ than those values for bias-uncorrected $\widehat{\gamma}$ when $c\le 0.15$.
When $c=0.3$ and $n=100$, they are a little lower than the target level.
On the other hand, when $n$ is fixed, the length of confidence interval of $\widehat{\gamma}$ increases as $c$ becomes larger.

{\renewcommand{\arraystretch}{1}
\begin{table}[!htbp]\centering
\caption{
The reported values are the coverage frequency ($\times 100\%$) for $\gamma_i$ for $i$ with bias-correction (uncorrected)/ bias ($\times 100$) / length of confidence interval ($\times 10$)
 /the frequency ($\times 100\%$) that the MLE did not exist ($\bs{\gamma}^*=(0.5, -0.5)^\top$).
}
\label{Table:gamma}
\begin{tabular}{lllll}
\hline
$n$     &   $c$    & $\gamma$        & $\epsilon_n=\log n/n^{1/6}$    & $\epsilon_n=\log n/n^{1/4}$   \\
\hline
$100$   & $0.05$   & $\gamma_1$      &$ 95.14 ( 93.69 ) / 1.03 / 0.13 / 0 $    &$ 94.84 ( 92.37 ) / 1.02 / 0.13 / 0 $ \\
        &          & $\gamma_2$      &$ 95.31 ( 93.62 ) / 1.01 / 0.12 / 0 $    &$ 95.01 ( 92.62 ) / 1.00 / 0.12 / 0 $  \\

        & $0.15$   & $\gamma_1$      &$ 94.65 ( 93.40 ) / 0.87 / 0.13 / 0 $ &$ 94.70 ( 93.18 ) / 0.87 / 0.13 / 0 $ \\
        &          & $\gamma_2$      &$ 94.71 ( 92.96 ) / 0.80 / 0.13 / 0 $ &$ 94.61 ( 92.57 ) / 0.80 / 0.13 / 0 $ \\

        & $0.3$    & $\gamma_1$      &$ 94.10 ( 93.74 ) / 0.37 / 0.15 / 0.08 $&$ 92.96 ( 92.17 ) / 0.37 / 0.15 / 0.40 $ \\
        &          & $\gamma_2$      &$ 93.87 ( 93.47 ) / 0.20 / 0.15 / 0.08 $&$ 93.01 ( 92.52 ) / 0.20 / 0.15 / 0.40 $ \\

        & $0.5 $   & $\gamma_1$      &$ 91.08 ( 92.72 ) / 0.75 / 0.20 / 28.16 $ &$ 89.70 ( 91.94 ) / 0.76 / 0.20 / 42.79 $  \\
        &          & $\gamma_2$      &$ 90.41 ( 93.19 ) / 1.08 / 0.19 / 28.16 $ &$ 89.06 ( 92.50 ) / 1.08 / 0.19 / 42.79 $  \\
\hline
$200$   & $0.05$   & $\gamma_1$      &$ 95.23 ( 93.72 ) / 0.51 / 0.06 / 0 $&$ 95.11 ( 93.50 ) / 0.51 / 0.06 / 0 $ \\
        &          & $\gamma_2$      &$ 95.31 ( 93.64 ) / 0.50 / 0.06 / 0 $&$ 95.29 ( 93.46 ) / 0.50 / 0.06 / 0 $  \\

        & $0.15$   & $\gamma_1$      &$ 95.38 ( 93.92 ) / 0.41 / 0.07 / 0 $&$ 95.14 ( 93.70 ) / 0.41 / 0.07 / 0 $  \\
        &          & $\gamma_2$      &$ 94.90 ( 93.11 ) / 0.37 / 0.06 / 0 $&$ 94.72 ( 92.85 ) / 0.37 / 0.06 / 0 $  \\

        & $0.3$    & $\gamma_1$      &$ 95.01 ( 94.80 ) / 0.08 / 0.08 / 0.01 $&$ 94.60 ( 94.30 ) / 0.08 / 0.08 / 0.04 $  \\
        &          & $\gamma_2$      &$ 94.45 ( 94.52 ) / 0.02 / 0.08 / 0.01 $&$ 94.10 ( 94.22 ) / 0.02 / 0.08 / 0.04 $  \\

        & $0.5 $   & $\gamma_1$      &$ 91.40 ( 93.64 ) / 0.64 / 0.11 / 16.65 $&$ 90.43 ( 93.00 ) / 0.64 / 0.11 / 33.72 $  \\
        &          & $\gamma_2$      &$ 90.10 ( 93.70 ) / 0.86 / 0.10 / 16.65 $&$ 89.56 ( 93.29 ) / 0.86 / 0.10 / 33.72 $  \\
\hline
\end{tabular}
\end{table}
}

\subsection{A real data example}

We use the Enron email dataset as an example analysis [\cite{Cohen2004}], available from \url{https://www.cs.cmu.edu/~enron/}.  
This dataset was released by William Cohen at Carnegie Mellon University and has been widely studied.
The Enron email data was originally acquired and made public by the Federal Energy Regulatory Commission during its investigation into fraudulent accounting practices.
Some of the emails have been deleted upon requests from affected employees.
However, the raw data is messy and needs to be cleaned before any analysis is conducted.
\cite{Zhou2007} applied data cleaning strategies to compile the Enron email dataset.
We use their cleaned data for the subsequent analysis.
The resulting data comprises $21,635$ messages sent between $156$ employees
with their covariate information. Thus, the corresponding
graph has multiple edges. We treat it as a simple undirected graph for our analysis, where
each edge denotes that there are at least one message between the corresponding two nodes.
We remove the isolated nodes ``32" and ``37" with zero degrees, where the estimators of the corresponding degree parameters do not exist.
This leaves a connected network with $154$ nodes and $1843$ edges. The minimum, $1/4$ quantile, median, $3/4$ quantile and maximum values of $d$ are $1$, $15$, $21$, $30$ and $88$, respectively.

\begin{figure}[!htp]
  \centering
  \subfigure{
  \centering
    \includegraphics[width=0.45\linewidth]{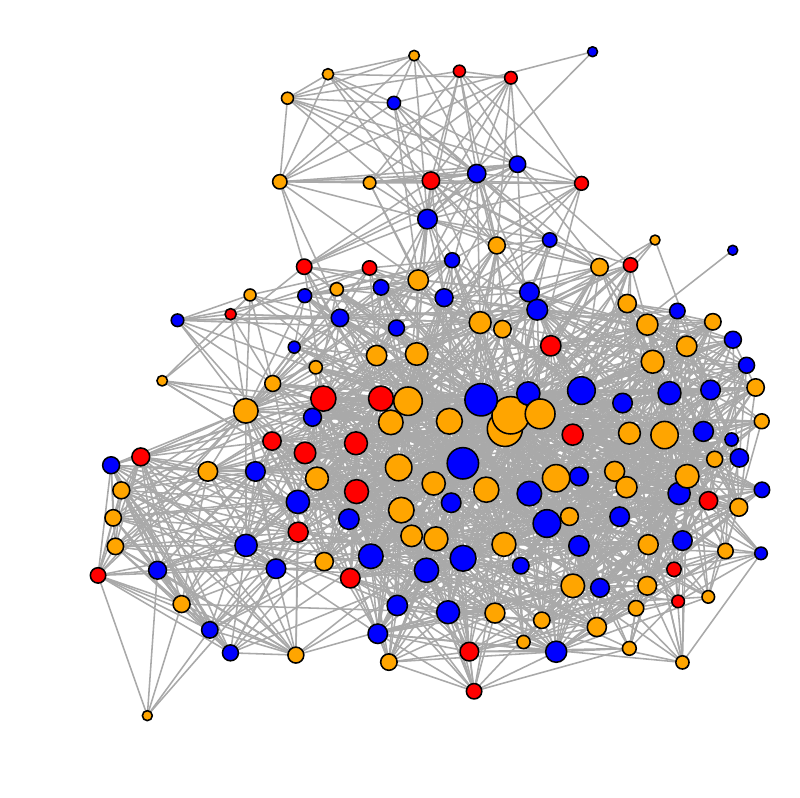}}
  \subfigure{
    \centering
    \includegraphics[width=0.45\linewidth]{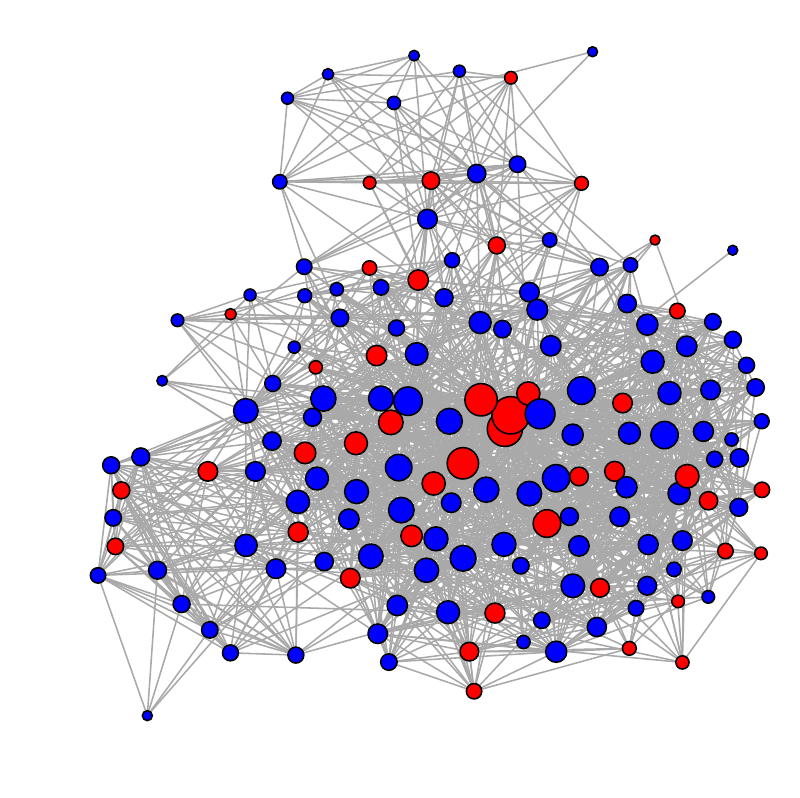}}
   \caption{Visualization of Enran email network among $154$ employees.
The vertex sizes are proportional to nodal degrees. The positions of the vertices are the same in two graphs.
For nodes with degrees less than $3$, we set their sizes the same (as a node with degrees 3).
In the left graph, the colors indicate different departments (red for legal and blue for trading and orange for other), while in the right graph, the colors represent different genders (blue for male and red for female).
}
\label{figure-data} 
\end{figure}

Each employee has three categorical variables: departments of these employees (Trading, Legal, Other),
the genders (Male, Female) and seniorities (Senior, Junior). 
We plot the network with individual departments and genders in Figure \ref{figure-data}.
We can see that the degrees exhibit a great variation across nodes and
it is not easy to judge homophic or heteriphic effects that require quantitative analysis.
The $3$-dimensional covariate vector $z_{ij}$ of edge $(i,j)$ is formed by using a homophilic matching function between these three covariates of two employees $i$ and $j$, i.e.,
if $x_{ik}$ and $x_{jk}$ are equal, then $z_{ijk}=1$; otherwise $z_{ijk}=-1$.

We evaluate how close the estimator $(\hat{\alpha}, \hat{\beta})$ is to the original
MLE $(\tilde{\alpha}, \tilde{\beta})$ fitted in the generalized $\beta$--model.
We chose two $\epsilon_n$ ($\log n/n^{1/6}$, $\log n/n^{1/4}$) and let $k_n=1$ as in simulations,
and repeated to release $d$ and $y$ using according to \eqref{eq-release}
$1,000$ times for each $\epsilon_n$. Then we computed the median private estimate and the upper ($97.5^{th}$) and the lower ($2.5^{th}$) quantiles.  The private estimate existed for each output.
The frequencies that the private estimate fails to exist are $33.7\%$ and $45.2\%$ for $\epsilon=\log n/n^{1/6}$ and $\epsilon=\log n/n^{1/4}$, respectively.
The results for the estimates $\widehat{\gamma}$ and $\widehat{\beta}$ are shown in Table \ref{Table-eran} and  Figure \ref{figure-enran}.
From Table \ref{Table-eran}, we can see that the median value of $\widehat{\gamma}$ is the same as the MLE and the MLE lies in the $95\%$ confidence interval.
The similar phenomenon can also be observed in Figure \ref{figure-enran}.
From this figure, we can see that the length of confidence interval of private estimates under $\epsilon_n= \log n/n^{1/6}$ are shorter than those under $\epsilon_n=\log n/n^{1/4}$.

{\renewcommand{\arraystretch}{1}
\begin{table}[h]\centering
\caption{The nonprivate MLE $\widehat{\gamma}_i^0$ of $\gamma_i$,
the median of private estimates of $\gamma_i$, the length of confidence interval for Enron email data.}
\label{Table-eran}
\vskip5pt
\begin{tabular}{lccl}
\hline
Covariate  &  $\hat{\gamma}_i^0$ & $\hat{\gamma}_{ i}$ & $95\%$ confidence interval  \\
\hline
&\multicolumn{3}{c}{$\epsilon_n= \log n/n^{1/6}$} \\
 \cline{2-4}
Department      &$ -0.016 $&$ -0.016 $&$ [ -0.083 , 0.033 ]$\\
Gender          &$ 0.063 $&$ 0.063 $&$ [ 0.005, 0.136 ]$\\
Seniority       &$ 0.032 $&$ 0.032 $&$ [ -0.022 , 0.085 ]$\\
\hline
&\multicolumn{3}{c}{$\epsilon_n= \log n/n^{1/4}$} \\
 \cline{2-4}
Department      &$ -0.016 $&$ -0.016 $&$ [ -0.083 , 0.033 ]$\\
Gender          &$ 0.063 $&$ 0.062 $&$ [ 0.004, 0.135 ]$\\
Seniority       &$ 0.032 $&$ 0.032 $&$ [ -0.021, 0.085]$\\
\hline
\end{tabular}
\end{table}
{\renewcommand{\arraystretch}{1}

\begin{figure}[!htb]
\centering
\includegraphics[ height=2.5in, width=5in, angle=0]{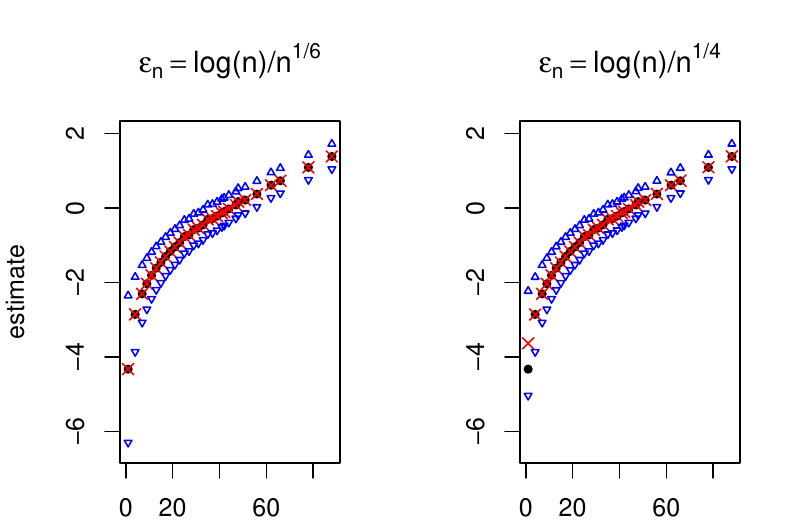}
\caption{The private estimate $\hat{\beta}$ in red color with the MLE in black color for the Enran email network.
The blue color indicates the upper and lower values of confidence interval.
}
\label{figure-enran}
\end{figure}

\section{Summary and discussion}
\label{section:sd}

We have present the $(k_n, \epsilon_n)$-edge-differentially private estimation for inferring the degree parameter
and homophily parameter in the generalized $\beta$--model.
We establish consistency of the estimator under several conditions and also derive its asymptotic normal distribution.
It is worth noting that the conditions imposed on $b_n$
imply the network density going to zero with a very slow rate.
When networks are very sparse, adding noises easily produces the outputs of negative degrees, which will lead to the
non-existence of the maximum likelihood estimator in the
covariate-adjusted $\beta$-model.
In addition, the conditions in Theorems  \ref{Theorem-central-a} and \ref{theorem-central-b} seem stronger than those needed for consistency.
Note that the asymptotic behavior of the estimator depends not only on $b_{n}$, but also on the configuration of the parameters.
It is of interest to see whether conditions for guaranteeing theoretical properties could be relaxed.

We use a generalized $\beta$--model to give a rigorous differential privacy analysis of networks with covariates.
It is notable that the assumption of the logistic distribution of an edge is not essential in our strategies for proofs.
Our these principled methods should be applicable to a class of network models beyond the covariate-adjusted $\beta$-model [\cite{wang-zhang-yan-2023}].
For instance, the developed two-stage Newton method to prove the consistency of the differentially private estimator
still works if the logistic distribution is replaced with the probit distribution.
Further, the edge independence assumption is not directly used in this method, which only plays a role to
derive the upper bound of $d$ and $y$. There are many tail probability inequalities in dependent random variables that could be applied
to network statistics with edge dependence situations.
We hope that the methods developed here can be further to be applied to edge-dependence network models.

\section{Appendix}
\label{section:appendix}
\subsection{Preliminaries}

In this section, we present three results that will be used in the proofs.
The first is on the approximation error of using $S$ to approximate the inverse of $V$ belonging to the matrix class
$\mathcal{L}_n(m, M)$,
where $V=(v_{ij})_{n\times n}$ and $S=\mathrm{diag}(1/v_{11}, \ldots, 1/v_{nn})$.
\cite{Yan:Zhao:Qin:2015} obtained the upper bound of the approximation error,
which has an order $n^{-2}$. The second is  a tight bound of $\| V^{-1} \|_\infty$ in \cite{hillar2012inverses}.
These two results are stated below as lemmas.

\begin{lemma}[\cite{Yan:Zhao:Qin:2015}] \label{pro:inverse:appro}  
If $V \in \mathcal{L}_n(m, M)$, then the following holds: 
\begin{equation*}\label{O-upperbound}
\|V^{-1} - S \|_{\max} 
 \le  \frac{1}{(n-1)^2} \left(\frac{ M }{2m^2}+\frac{nM^2}{2(n-2)m^3}+\frac{3n-2}{2nm}\right)=O(\frac{M^2}{n^2m^3}),
\end{equation*}
where $\|A \|_{\max}=\max_{i,j} |A_{ij}| $ for a general matrix $A$.
\end{lemma}

\begin{lemma}[\cite{hillar2012inverses}]
\label{lemma-tight-V}
For $V\in \mathcal{L}_n(m, M)$, when $n\ge 3$, we have
\[
\frac{1}{2M(n-1)} \le \|V^{-1}\|_\infty \le \frac{3n-4}{2m(n-1)(n-2)}.
\]
\end{lemma}

Let $F(x): \R^n \to \R^n$ be a function vector on $x\in\R^n$. We say that a Jacobian matrix $F^\prime(x)$ with $x\in \R^n$ is Lipschitz continuous on a convex set $D\subset\R^n$ if
for any $x,y\in D$, there exists a constant $\lambda>0$ such that
for any vector $v\in \R^n$ the inequality
\begin{equation*}
\| [F^\prime (x)] v - [F^\prime (y)] v \|_\infty \le \lambda \| x - y \|_\infty \|v\|_\infty
\end{equation*}
holds.
We will use the Newton iterative sequence to establish the existence and consistency of the differentially private estimator.
\cite{Gragg:Tapia:1974} gave the optimal error bound for the Newton method under the Kantovorich conditions
[\cite{Kantorovich1948Functional}]. We only show partial results here that are enough for our applications.

\begin{lemma}[\cite{Gragg:Tapia:1974}]\label{lemma:Newton:Kantovorich}
Let $D$ be an open convex set of $\R^n$ and $F:D \to \R^n$ be Fr\'{e}chet differentiable on $D$
with a Jacobian $F^\prime(x)$ that is Lipschitz continuous on $D$ with Lipschitz coefficient $\lambda$.
Assume that $x_0 \in D$ is such that $[ F^\prime (x_0) ]^{-1} $ exists,
\[
\| [ F^\prime (x_0 ) ]^{-1} \|  \le \aleph,~~ \| [ F^\prime (x_0) ]^{-1} F(x_0) \| \le \delta, ~~ h= 2 \aleph \lambda \delta \le 1,
\]
and
\[
B(x_0, t^*) \subset D, ~~ t^* = \frac{2}{h} ( 1 - \sqrt{1-h} ) \delta = \frac{ 2\delta }{ 1 + \sqrt{1-h} }.
\]
Then: (1) The Newton iterations $x_{k+1} = x_k - [ F^\prime (x_k) ]^{-1} F(x_k)$ exist and $x_k \in B(x_0, t^*) \subset D$ for $k \ge 0$. (2)
$x^* = \lim x_k$ exists, $x^* \in \overline{ B(x_0, t^*) } \subset D$ and $F(x^*)=0$.
\end{lemma}

\subsection{Error bound between $\widehat{\beta}_\gamma$ and $\beta^*$}
We will use the Newton method to derive the error bound between $\widehat{\beta}_\gamma$ and $\beta^*$ through
verifying the Kantororich conditions in Lemma \ref{lemma:Newton:Kantovorich}.
The Kantororich conditions require the Lipschitz continuous of $F^\prime_{\gamma}(\beta)$ and the upper bounds of
$F_\gamma(\beta^*)$. We need three lemmas below, whose proofs are in supplementary material.

\begin{lemma}\label{lemma:lipschitz-c}
For any given $\gamma$,
the Jacobian matrix $F'_\gamma( x )$ of $F_\gamma(x)$ on $x$ is Lipschitz continuous on $\R^n$ with the Lipschitz coefficient  $  (n-1)$.
\end{lemma}

\begin{lemma}\label{lemma-diff-F-Q}
The tail probabilities of $\|d - \E d \|_\infty$ and $\| y - \E y\|_\infty$ are given below:
\begin{eqnarray*}
\P\Bigg(\|d - \E d \|_\infty \ge \sqrt{n\log n} \Bigg)
& \le &  \frac{2}{n }, \\
 \P\left(\|y - \E y \|_\infty \ge  2z_* \sqrt{n(n-1)\log n/2}  \right) & \le & \frac{ 2p}{n^2}.
\end{eqnarray*}
\end{lemma}

\begin{lemma}\label{lemma-max-ei}
For any given $c>0$, we have
\begin{equation*}\label{ineq:e}
\P( \max_{i=1, \ldots, n} |\xi_i|  > c )  < n \exp( -\frac{c\varepsilon_n}{4k_n} ),~~
\P( \max_{i=1, \ldots, p} |\eta_i|  > c ) <  p \exp( - \frac{ c \varepsilon_n}{ p z_* k_n}).
\end{equation*}
\end{lemma}

Recall that
\[
\tilde{\varepsilon}_n = 1 + \frac{ 4 k_n }{ \varepsilon_n } \sqrt{ \frac{\log n}{n} }.
\]
We state the error bound between $\widehat{\beta}_\gamma$ and $\beta^*$ below.

\begin{lemma}\label{lemma-hatbeta-con}
Assume $\epsilon_{n2}=O(  (\log n)^{1/2} n^{-1/2})$ and $\gamma\in B(\gamma^*, \epsilon_{n2})$.
If $b_n^2\tilde{\varepsilon}_n= o( (n/\log n)^{1/2} )$,
then with probability approaching one,
 $\widehat{\beta}_\gamma$ exists and satisfies
\[
\| \widehat{\beta}_\gamma - \beta^* \|_\infty = O_p\left(  b_n \tilde{\varepsilon}_n \sqrt{\frac{\log n}{n}} \right) = o_p(1).
\]
\end{lemma}

\begin{proof}[Proof of Lemma \ref{lemma-hatbeta-con}]
We will derive the error bound between $\widehat{\beta}_\gamma$ and $\beta^*$ through constructing
the Newton iterative sequence $\beta^{(n+1)}= \beta^{(n)} - F_\gamma^\prime (\beta^{(n)}) F_\gamma (\beta^{(n)})$,
where we choose $\beta^*$ as the starting point $\beta^{(0)}:=\beta^*$. To this end, it is sufficient to verify the Kantovorich conditions in Lemma \ref{lemma:Newton:Kantovorich}.
Note that $F'_{\gamma}(\beta) \in \mathcal{L}_n(b_{n}^{-1}, 1/4)$
when $\beta\in B(\beta^*, \epsilon_{n1})$ and $\gamma \in B(\gamma^*, \epsilon_{n2})$,
where $\epsilon_{n1}$ is a small positive number and
$\epsilon_{n2}= O( (\log n/n)^{1/2})$.
The following calculations are based on the event $E_n$:
\[
E_n = \{d: \| \tilde{d} - \E d \|_\infty \le  \tilde{\varepsilon}_n (n\log n)^{1/2}  \}.
\]

Let $V=(v_{ij})= \partial F_{\gamma}(\beta^*)/\partial \beta^\top$ and $S=\mathrm{diag}(1/v_{11}, \ldots, 1/v_{nn})$.
By Lemma \ref{lemma-tight-V}, we have
$\aleph =\|V^{-1}\|_\infty = O( b_n/n  )$.
Recall that $F_{\gamma^*}(\beta^*) =  \E d - \tilde{d} $. 
Note that the dimension $p$ of $\gamma$  is a fixed constant.
If $\epsilon_{n2}=O( (\log n)^{1/2} n^{-1/2})$, by the mean value theorem and the event $E_n$, we have
\begin{eqnarray*}
\| F_\gamma(\beta^*) \|_\infty & \le & \| \tilde{d} - \E d \|_\infty + \max_i | \sum\nolimits_{j\neq i} [\mu_{ij}(\beta^*, \gamma) - \mu_{ij}(\beta^*, \gamma^*)] |  \\
& \le & O(  \tilde{\varepsilon}_n (n\log n)^{1/2}) + \max_i \sum_{j\neq i} |\mu_{ij}^\prime( \beta^*, \bar{\gamma})| | z_{ij}^\top ( \gamma - \gamma^*)| \\
& \le & O(  \tilde{\varepsilon}_n (n\log n)^{1/2}) + O( n \cdot (\log n)^{1/2} n^{-1/2} ) \\
& = & O(  \tilde{\varepsilon}_n (n\log n)^{1/2}).
\end{eqnarray*}
Repeatedly utilizing  Lemma \ref{lemma-tight-V}, we have
\begin{eqnarray*}
\delta=\| [F'_\gamma(\beta^*)]^{-1}F_\gamma(\beta^*) \|_\infty = \| [F'_\gamma(\beta^*)]^{-1}\|_\infty \|F_\gamma(\beta^*) \|_\infty
=O\left(  \tilde{\varepsilon}_n b_n \sqrt{\frac{\log n}{n}} \right)
\end{eqnarray*}
By Lemma \ref{lemma:lipschitz-c}, $F_\gamma(\beta)$ is Lipschitz continuous with Lipschitz coefficient $\lambda=n-1$.
Therefore, if $\tilde{\varepsilon}_n b_n^2=o( (n/\log n)^{1/2} )$, then
\begin{eqnarray*}
h =2\aleph \lambda \delta & = & O(\frac{b_n}{n})\times O( n )
\times O(  \tilde{\varepsilon}_n b_n \sqrt{\frac{\log n}{n}} ) \\
& = & O\left(  \tilde{\varepsilon}_n b_n^2 \sqrt{ \frac{\log n}{n} } \right) =o(1).
\end{eqnarray*}
The above arguments verify the Kantovorich conditions.
By Lemma \ref{lemma:Newton:Kantovorich}, it yields that
\begin{equation}
\label{eq-hatbeta-upper}
\| \widehat{\beta}_\gamma - \beta^* \|_\infty = O\left( \tilde{\varepsilon}_n b_n \sqrt{\frac{\log n}{n}} \right).
\end{equation}

To finish the proof, it is left to show $\P (E_n^c) \to 0$. Note that $\tilde{d}_i = d_i + \xi_i$.
By Lemmas \ref{lemma-diff-F-Q} and \ref{lemma-max-ei},
it can be verified as follows:
\begin{eqnarray*}
\P (E_n^c) &=& \P( \| \tilde{d} - \E d \|_\infty > \tilde{\varepsilon}_n \sqrt{n\log n}) \\
  & \le & \P( \| d - \E d \|_\infty > \sqrt{n\log n} )
+\P( \max_{i=1, \ldots, n} \xi_i > (\tilde{\varepsilon}_n-1) \sqrt{n\log n}) \\
&\le & \frac{2}{n} \to 0.
\end{eqnarray*}
It completes the proof.
\end{proof}

\subsection{Proof of Theorem \ref{Theorem:con}}

To show Theorem \ref{Theorem:con}, we need three lemmas below.

\begin{lemma}\label{lemma-Q-Lip}
Let $D=B(\gamma^*, \epsilon_{n2}) (\subset \R^{p})$ be an open convex set containing the true point $\gamma^*$.
If $\| \tilde{d} - \E d\|_\infty= O( \tilde{\epsilon}_n (n\log n)^{1/2})$, then
$ Q_c(\gamma)$ is Lipschitz continuous on $D$ with the Lipschitz coefficient  $n^2 b_{n}^{-3}$.
\end{lemma}

\begin{lemma}
\label{lemma-asym-expansion-beta}
Write $\widehat{\beta}^*$ as $\widehat{\beta}_{\gamma^*}$
and $V=\partial F(\beta^*, \gamma^*)/\partial \beta^\top$.
If $b_n^2\tilde{\varepsilon}_n= o( (n/\log n)^{1/2} )$, then
$\widehat{\beta}^*$ has the following expansion:
\begin{equation}\label{result-lemma2}
\widehat{\beta}^* - \beta^* =  V^{-1} F(\beta^*, \gamma^*) +V^{-1}R,
\end{equation}
where $R=(R_1, \ldots, R_{n})^\top$ is the remainder term and
\[
\left\| V^{-1}  R \right \|_\infty = O_p( \frac{b_n^3 \tilde{\varepsilon}_n^2  \log n}{n}  ).
\]
\end{lemma}

\begin{lemma}
\label{lemma-order-Q-beta}
If $b_n^2\tilde{\varepsilon}_n= o( (n/\log n)^{1/2} )$,
for any $\beta \in B( \beta^*, \epsilon_{n1})$ and $\gamma \in B( \gamma^*, \epsilon_{n2})$,  then we have
\[
 \| \frac{\partial Q (\beta, \gamma)}{ \partial \beta^\top }
  (\widehat{\beta}_{\gamma}-\beta^*) \|_\infty =
   O_p(  \tilde{\varepsilon}_n b_n^3 n\log n ).
\]
\end{lemma}

Now we are ready to prove  Theorem \ref{Theorem:con}.

\begin{proof}[Proof of Theorem \ref{Theorem:con}]
We construct the Newton iterative sequence to show the consistency. It is sufficient to verify the
Kantovorich conditions in Lemma \ref{lemma:Newton:Kantovorich}.
In the Newton method, we set $\gamma^*$ as the initial point $\gamma^{(0)}$ and $\gamma^{(k+1)}=\gamma^{(k)} - [Q_c^\prime(\gamma^{(k)})]^{-1}Q_c(\gamma^{(k)})$.

The following calculations are based on the event $E_n$ that for $\gamma\in B( \gamma^*, \epsilon_{n2})$,
$\widehat{\beta}_\gamma$ exists and satisfies
\[
\| \widehat{\beta}_\gamma - \beta^* \|_\infty = O\left( \tilde{\varepsilon}_n b_n \sqrt{\frac{\log n}{n}} \right). 
\]
This shows that $\widehat{\beta}_{\gamma^{(0)}}$ exists such that $Q_c(\gamma^{(0)})$ and $Q_c^\prime(\gamma^{(0)})$ are well defined.
This in turn shows that in every iterative step, $\gamma^{(k+1)}$ exists as long as $\gamma^{(k)}$ exists.

Recall the definition of $Q_c(\gamma)$ and $Q(\beta, \gamma)$ in \eqref{definition-Q} and \eqref{definition-Qc}.
By Lemmas \ref{lemma-diff-F-Q} and  \ref{lemma-max-ei}, we have
\begin{eqnarray*}
\| Q(\beta^*, \gamma^*) \|_\infty & = & \| \E y - y \|_\infty   + \|\eta \|_\infty \\
& = & O_p( z_*n(\log n)^{1/2} ) + O_p( (pk_n/\varepsilon_n) \log n ) \\
& = & O_p( (z_* + \frac{ pk_n (\log n)^{1/2} }{ n\varepsilon_n} ) n(\log n)^{1/2} ).
\end{eqnarray*}
By the mean value theorem and Lemma \ref{lemma-order-Q-beta}, we have
\[
\|Q(\widehat{\beta}^*, \gamma^*)-Q(\beta^*, \gamma^*)\|_\infty = \| \frac{\partial Q (\bar{\beta}, \gamma)}{ \partial \beta^\top }
  (\widehat{\beta}_{\gamma}-\beta^*) \|_\infty = O_p(  \tilde{\varepsilon}_n b_n^3 n\log n ) ).
\]
Then it follows that
\begin{eqnarray*}
\|Q_c(\gamma^*)\|_\infty  & \le &  \|Q(\beta^*, \gamma^*)\|_\infty + \|Q(\widehat{\beta}^*, \gamma^*)-Q(\beta^*, \gamma^*)\|_\infty \\
& = & O_p\left( ( b_n^3 +  (k_n/\varepsilon_n)b_n^3+ z_*/(\log n)^{1/2}) n\log n )  \right) :=O_p( \tau_n n\log n).
\end{eqnarray*}
By Lemma \ref{lemma-Q-Lip}, $Q^\prime_c(\gamma)$ is Lipschitz continuous with $\lambda=n^2  b_{n}^3 $.
Note that
$\aleph=\| [Q_c^\prime(\gamma^*)]^{-1} \|_\infty = O ( \rho_n n^{-2})$.
Thus,
\[
\delta = \| [Q_c^\prime(\gamma^*)]^{-1} Q_c(\gamma^*) \|_\infty =  O_p\left(
 \frac{\rho_n \tau_n \log n }{ n } \right  ).
\]
 As a result, if $b_n^3\rho_n^2 \tau_n = o( n/\log n)$, then
\[
h=2\aleph \lambda \delta =
O_p( \frac{ \rho_n}{n^2} \cdot n^2 b_n^3 \cdot \rho_n \tau_n \cdot \frac{\log n}{n} )=o_p(1).
\]
By Lemma \ref{lemma:Newton:Kantovorich}, with probability approaching one, the limiting point of the sequence $\{\gamma^{(k)}\}_{k=1}^\infty$ exists denoted by $\widehat{\gamma}$, and satisfies
\[
\| \widehat{\gamma} - \gamma^* \|_\infty =O_p\left(
 \frac{ \rho_n \tau_n \log n}{ n} \right  ).
\]
At the same time, by Lemma \ref{lemma-hatbeta-con}, $\widehat{\beta}_{\widehat{\gamma}}$ exists, denoted by $\widehat{\beta}$.
The limiting points $(\widehat{\beta}, \widehat{\gamma})$ satisfies the equation \eqref{eq:moment:dp}.
It completes the proof.
\end{proof}

\subsection{Proofs for Theorem \ref{Theorem-central-a}}

\begin{proof}[Proof of Theorem \ref{Theorem-central-a}]

To simplify notations, write $\widehat{\pi}_{ij}=\widehat{\beta}_i+\widehat{\beta}_j+z_{ij}^\top \widehat{\gamma}$,
$\pi_{ij}^*=\beta_i^* + \beta_j^* + z_{ij}^\top \gamma^*$, $\mu_{ij}^\prime = \mu^\prime(\pi_{ij}^*)$ and
\[
V= \frac{ \partial F(\beta^*, \gamma^*)}{\partial \beta^\top}, ~~ V_{\beta\gamma} = \frac{ \partial F(\beta^*, \gamma^*)}{\partial \gamma^\top}.
\]
By a second order Taylor expansion, we have
\begin{equation}
\label{equ-Taylor-exp}
\mu( \widehat{\pi}_{ij} ) - \mu(\pi_{ij}^*)
= \mu_{ij}^\prime (\widehat{\beta}_i-\beta_i)+\mu_{ij}^\prime (\widehat{\beta}_j-\beta_j) + \mu_{ij}^\prime z_{ij}^\top ( \widehat{\gamma} - \gamma)
+ g_{ij},
\end{equation}
where
\[
g_{ij}= \frac{1}{2} \begin{pmatrix}
\widehat{\beta}_i-\beta_i^* \\
\widehat{\beta}_j-\beta_j^* \\
\widehat{\gamma} - \gamma^*
\end{pmatrix}^\top
\begin{pmatrix}
\mu^{\prime\prime}_{ij}( \tilde{\pi}_{ij} ) & -\mu^{\prime\prime}_{ij}( \tilde{\pi}_{ij} )
& \mu^{\prime\prime}_{ij}( \tilde{\pi}_{ij} ) z_{ij}^\top \\
-\mu^{\prime\prime}_{ij}( \tilde{\pi}_{ij} ) & \mu^{\prime\prime}_{ij}( \tilde{\pi}_{ij} )
& -\mu^{\prime\prime}_{ij}( \tilde{\pi}_{ij} ) z_{ij}^\top \\
\mu^{\prime\prime}_{ij}( \tilde{\pi}_{ij} ) z_{ij}^\top
& -\mu^{\prime\prime}_{ij}( \tilde{\pi}_{ij} ) z_{ij}^\top & \mu^{\prime\prime}_{ij}( \tilde{\pi}_{ij} ) z_{ij}z_{ij}^\top
\end{pmatrix}
\begin{pmatrix}
\widehat{\beta}_i-\beta_i^* \\
\widehat{\beta}_j-\beta_j^* \\
\widehat{\gamma} - \gamma^*
\end{pmatrix},
\]
and $\tilde{\pi}_{ij}$ lies between $\pi_{ij}^*$ and $\widehat{\pi}_{ij}$.
Recall that $z_*:= \max_{i,j} \| z_{ij} \|_\infty$.
Since $|\mu^{\prime\prime}(\pi_{ij})|\le 1/4$ (see \eqref{eq-mu-d-upper}), we have
\begin{equation*}
\begin{array}{rcl}
|g_{ij}| & \le &  \| \widehat{\beta} - \beta^*\|_\infty^2 + \tfrac{1}{2}\| \widehat{\beta} - \beta^*\|_\infty \| \widehat{\gamma}-\gamma^* \|_1 z_* + \tfrac{1}{4} \| \| \widehat{\gamma}-\gamma^* \|_1^2 z_*^2 \\
& \le & \tfrac{1}{2}[4 \| \widehat{\beta} - \beta^*\|_\infty^2+  \| \widehat{\gamma}-\gamma^* \|_1^2 z_*^2].
\end{array}
\end{equation*}
Let  $g_i=\sum_{j\neq i}g_{ij}$ and $g=(g_1, \ldots, g_n)^\top$.
If  $\rho_n^2 \tau_n^2 = o( n/\log n)$, by Theorem \ref{Theorem:con}, we have
\begin{equation}
\label{inequality-gij}
\max_{i=1, \ldots, n} |g_i| \le n\max_{i,j} |g_{ij}| = O_p\left(
\left( b_n^2 \tilde{\varepsilon}_n^2 + \frac{ \rho_n^2 \tau_n^2 \log n}{n}\right) \log n  \right)
= O_p\left( b_n^2 \tilde{\varepsilon}_n^2 \log n \right).
\end{equation}
By writing \eqref{equ-Taylor-exp} into a matrix form, we have
\[
\tilde{d} - \E d = V(\widehat{\beta} - \beta^*) + V_{\beta\gamma} (\widehat{\gamma}-\gamma^*) + g ,
\]
which is equivalent to
\begin{equation}
\label{expression-beta}
\widehat{\beta} - \beta^* = V^{-1}(d - \E d) + V^{-1}V_{\beta\gamma} (\widehat{\gamma}-\gamma^*) + V^{-1} g + V^{-1} \xi.
\end{equation}
We bound the last three remainder terms in the above equation as follows.
Let $W=V^{-1} - S$.
Note that $(S g)_i =  g_i/v_{ii}$ and $(n-1)b_{n}^{-1} \le v_{ii} \le (n-1)/4$.
By Lemma \ref{lemma-tight-V} and inequality \eqref{inequality-gij}, we have
\begin{eqnarray}
\label{eq-Vg-upper}
\| V^{-1} g \|_\infty \le \|V^{-1}\|_\infty \|g \|_\infty
= O( \frac{b_n}{n} \times    b_{n}^2 \tilde{\varepsilon}_n^2 \log n ) = O_p(  \frac{ b_n^3 \tilde{\varepsilon}_n^2 \log n}{ n } ).
\end{eqnarray}

Note that the $i$th row of $V_{\beta\gamma}$ is $\sum_{j=1,j\neq i}^n \mu^\prime_{ij} z_{ij}^\top$.
By Theorem \ref{Theorem:con}, we have
\[
\| V_{\beta\gamma}(\widehat{\gamma}-\gamma^*) \|_\infty \le (n-1) z_*\|\widehat{\gamma}-\gamma^*\|_1
=  O_p( z_* \rho_n \tau_n \log n ).
\]
By Lemma \ref{pro:inverse:appro}, we have
\begin{equation}
\label{equ-theorem3-3}
\begin{array}{rcl}
\|V^{-1}V_{\beta\gamma} (\widehat{\gamma}-\gamma^*)\|_\infty    \le
\|V^{-1}\|_\infty \|V_{\gamma\beta} (\widehat{\gamma}-\gamma^*)\|_\infty
= O_p( \frac{z_* \rho_n \tau_n b_n   \log n}{n }).
\end{array}
\end{equation}
By Lemma \ref{lemma-max-ei},
\[
P( \| \xi \|_\infty > 8 (k_n/\varepsilon_n) \log n ) \le n \times \exp( - 8 (k_n/\varepsilon_n) \log n \times 4(\varepsilon_n/k_n)  ) = \frac{1}{n},
\]
such that
\[
\| \xi \|_\infty = O_p( (k_n/\varepsilon_n) \log n).
\]
Thus, it yields
\begin{equation}\label{eq-Vxi-upper}
\| V^{-1} \xi \|_\infty  \le  \| V^{-1} \|_\infty \| \xi \|_\infty = O_p \left( b_n (k_n/\varepsilon_n ) \frac{\log n}{n} \right).
\end{equation}
By combining \eqref{expression-beta}, \eqref{eq-Vg-upper}, \eqref{equ-theorem3-3} and \eqref{eq-Vxi-upper}, it yields
\[
\widehat{\beta}_i - \beta^*_i = [V^{-1}(d - \E d)]_i
+  O_p( \frac{( b_n^3 \tilde{\varepsilon}_n^2 + z_* \rho_n \tau_n  b_n) \log n}{n}).
\]

Let $W=V^{-1}-S$ and $U=\mathrm{Cov}( W (d- \E d) )$. It is easy to verify that
\[
U= V^{-1} - S  - S( I_n - VS)
\]
and
\[
[S(I_n - VS)]_{ij} = \frac{ ( \delta_{ij} - 1) v_{ij} }{ v_{ii} v_{jj} }.
\]
By Lemma \ref{pro:inverse:appro}, we have
\[
\| U \|_{\max} = O( b_n^3 n^{-2}).
\]
Therefore,
\begin{equation}
\label{equ-theorem3-dd}
 [W(d - \E d)]_i = O_p(  b_n^{3/2} n^{-1} ).
\end{equation}
Consequently, by combining \eqref{expression-beta}, \eqref{eq-Vg-upper}, \eqref{equ-theorem3-3} and \eqref{equ-theorem3-dd},
we have
\[
\widehat{\beta}_i - \beta^*_i = \frac{ d_i - \E d_i}{v_{ii}} +  + O_p( \frac{( b_n^3 \tilde{\varepsilon}_n^2 + z_* \rho_n \tau_n  b_n) \log n}{n}).
\]
Therefore, Theorem \ref{Theorem-central-a} immediately follows from Proposition  \ref{pro:central:poisson}.
\end{proof}

\subsection{Proof of Theorem~\ref{theorem-central-b}}

\begin{proof}[Proof of Theorem \ref{theorem-central-b}]
Assume that the conditions in Theorem \ref{Theorem:con} hold.
A mean value expansion gives
\[
 Q_c( \widehat{\gamma} ) - Q_c(\gamma^*) =  \frac{\partial Q_c(\bar{\gamma}) }{ \partial \gamma^\top }  (\widehat{\gamma}-\gamma^*),
\]
where $\bar{\gamma}$ lies between $\gamma^*$ and $\widehat{\gamma}$.
By noting that $Q_c( \widehat{\gamma} )=0$, we have
\[
\sqrt{N}(\widehat{\gamma} - \gamma^*) = -
\Big[ \frac{1}{N}  \frac{\partial Q_c(\bar{\gamma}) }{ \partial \gamma^\top } \Big]^{-1}
\times \frac{1}{\sqrt{N}}¡¡Q_c(\gamma^*). 
\]
Note that the dimension of $\gamma$ is fixed. By Theorem \ref{Theorem:con} and \eqref{equation-H-appro}, we have
\[
\frac{1}{N}  \frac{\partial Q_c(\bar{\gamma}) }{ \partial \gamma^\top }
\stackrel{p}{\to } \bar{H}=\lim_{N\to\infty} \frac{1}{N}H(\beta^*, \gamma^*).
\]
Write $\widehat{\beta}^*$ as $\widehat{\beta}(\gamma^*)$ for convenience. Let $\bar{Q}(\beta, \gamma)= Q (\beta, \gamma) - \eta$
and $\bar{Q}_c(\beta,\gamma) = Q_c(\beta,\gamma) - \eta$.
Note that $\eta $ is a Laplace random vector.
By Lemma 10, if $k_n/\epsilon_n = o( n/\log n)$, then
\[
\frac{ \|\eta\|_\infty }{N^{1/2}} =O_p\left( \frac{p(k_n/\epsilon_n) \log n}{n} \right) = o_p(1).
\]
 Therefore,
\begin{equation}\label{eq:theorem4:aa}
\sqrt{N} (\widehat{\gamma} - \gamma^*) = - \bar{H}^{-1} \cdot \frac{1}{\sqrt{N}} \bar{Q}( \widehat{\beta}^*, \gamma^*)  + o_p(1).
\end{equation}
By applying a third order Taylor expansion to $\bar{Q}( \widehat{\beta}^*, \gamma^*)$, it yields
\begin{equation}\label{eq:gamma:asym:key}
\frac{1}{\sqrt{N}}  \bar{Q}(\widehat{\beta}^*, \gamma^*) = S_1 + S_2 + S_3,
\end{equation}
where
\begin{equation*}
\begin{array}{l}
S_1  =  \frac{1}{\sqrt{N}}  \bar{Q}(\beta^*, \gamma^* )
+ \frac{1}{\sqrt{N}}
\Big[\frac{\partial  \bar{Q}(\beta^*, \gamma^* ) }{\partial \beta^\top } \Big]( \widehat{\beta}^* - \beta^* ), \\
S_2  =   \frac{1}{2\sqrt{N}} \sum_{k=1}^{n} \Big[( \widehat{\beta}_k^* - \beta_k^* )
\frac{\partial^2 \bar{Q}(\beta^*, \gamma^* ) }{ \partial \beta_k \partial \beta^\top }
\times ( \widehat{\beta}^* - \beta^* ) \Big],  \\
S_3  =  \frac{1}{6\sqrt{N}} \sum_{k=1}^{n} \sum_{l=1}^{n} \{ (\widehat{\beta}_k^* - \beta_k^*)(\widehat{\beta}_l^* - \beta_l^*)
\Big[   \frac{ \partial^3 \bar{Q}(\bar{\beta}^*, \gamma^*)}{ \partial \beta_k \partial \beta_l \partial \beta^\top } \Big]
(\widehat{\beta}^*  - \beta^* )\},
\end{array}
\end{equation*}
and $\bar{\beta}^*=t\beta^*+(1-t)\widehat{\beta}^*$ for some $t\in(0,1)$.
We will show that (1) $S_1$ asymptotically follows a multivariate normal distribution;
(2) $S_2$ is a bias term; (3) $S_3$ is an asymptotically negligible remainder term.
Specifically, they are accurately characterized as follows:
\begin{eqnarray*}
S_1 & = &  \frac{1}{\sqrt{N}} \sum_{j< i} s_{ij}(\beta^*, \gamma^*) + O_p(\frac{ b_n^3 \tilde{\varepsilon}_n^3 z_* \log n}{n}), \\
S_2 & = & \sum_{k=1}^{n} \frac{ \sum_{j\neq k} \mu_{kj}^{\prime\prime}(\beta^*, \gamma^*) z_{kj} }{ v_{kk} }
+ O_p\left( \frac{b_n^4 \tilde{\varepsilon}_n^3 z_* (\log n)^{1/2} }{ n^{1/2}}  \right), \\
\|S_3\|_\infty & = &  O_p( \frac{ (\log n)^{3/2}b_n^3 \tilde{\varepsilon}_n^3 }{n^{1/2}} ).
\end{eqnarray*}
We defer the proofs of the above equations to supplementary material.
Substituting the above equations into \eqref{eq:theorem4:aa} then gives
\[
\sqrt{N}(\widehat{\gamma}- \gamma^*) = - \bar{H}^{-1} B_* + \bar{H}^{-1} \times \frac{1}{\sqrt{N}}  \sum_{i< j}
s_{ij} (\beta^*, \gamma^*) + O_p\left( \frac{b_n^4 \tilde{\varepsilon}_n^3 z_* (\log n)^{3/2} }{ n^{1/2}}  \right).
\]
If $b_n^4 \tilde{\varepsilon}_n^3 z_*=o(n^{1/2}/(\log n)^{3/2})$,
then Theorem \ref{theorem-central-b} immediately follows from Proposition \ref{pro:th4-b}.

\end{proof}

\setlength{\itemsep}{-1.5pt}
\setlength{\bibsep}{0ex}
\bibliography{reference3}
\bibliographystyle{apalike}

\end{document}